\documentclass[aps,prx,twocolumn,superscriptaddress,floatfix,citeautoscript,preprintnumbers,nofootinbib]{revtex4-2}
\usepackage[utf8]{inputenc}
\pdfoutput=1

\usepackage{graphicx}% Include figure files
\usepackage{dcolumn}% Align table columns on decimal point
\usepackage{bm}% bold math

\usepackage[font=small,labelfont=bf,justification=justified]{caption} % Para hacer más pequeños y en negrita los pies de figura
\usepackage{dsfont} % Para poner el \mathds{1} guay de la identidad.
\usepackage{physics} % Para usar braket notation.
\usepackage{blindtext}
\usepackage{amsfonts, amsthm, microtype, mathrsfs, bbm}

\usepackage[colorlinks,allcolors=blue]{hyperref}
\hypersetup{breaklinks=true}
\usepackage{mathtools, enumerate}

\usepackage{tabularx}
\usepackage{makecell}
\usepackage{stfloats}
\usepackage{float}
\usepackage{here}
\usepackage{multirow}

\usepackage{balance}

\usepackage{placeins}
\usepackage{thmtools, thm-restate}

\usepackage{soul}

\usepackage{adjustbox}

\newtheorem*{definition*}{Definition}
\newtheorem{theorem}{Theorem}
\newtheorem*{theorem*}{Theorem}
\newtheorem{lemma}{Lemma}
\newtheorem{proposition}[theorem]{Proposition}
\newtheorem*{proposition*}{Proposition}

\newtheorem*{corollary*}{Corollary}
\theoremstyle{definition}

\newtheorem{innercustomdef}{Definition}
\newenvironment{mydef}[1]
  {\renewcommand\theinnercustomdef{#1}\innercustomdef}
  {\endinnercustomdef}

\theoremstyle{plain}  
\newtheorem{innercustomthm}{Theorem}
\newenvironment{mythm}[1]
  {\renewcommand\theinnercustomthm{#1}\innercustomthm}
  {\endinnercustomthm}

\newtheorem{innercustomcly}{Corollary}
\newenvironment{mycly}[1]
  {\renewcommand\theinnercustomcly{#1}\innercustomcly}
  {\endinnercustomcly}

\newtheorem{innercustomprop}{Proposition}
\newenvironment{myprop}[1]
  {\renewcommand\theinnercustomprop{#1}\innercustomprop}
  {\endinnercustomprop}

\usepackage{tikz}
\usetikzlibrary{decorations.pathreplacing, calligraphy, decorations.markings}
\usetikzlibrary{arrows, automata, positioning}

\definecolor{purple}{HTML}{E5E3F5}
\definecolor{amaranth}{HTML}{F5CFD8}
\definecolor{yellow}{HTML}{FFEAB6}
\definecolor{pink}{HTML}{FFDBDA}

\usepackage{tikz}
\usetikzlibrary{decorations.pathreplacing, calligraphy, decorations.markings}

\definecolor{purple}{HTML}{E5E3F5}
\definecolor{amaranth}{HTML}{F5CFD8}
\definecolor{yellow}{HTML}{FFEAB6}
\definecolor{pink}{HTML}{FFDBDA}

\definecolor{darkpurple}{HTML}{3D348B}
\definecolor{darkamaranth}{HTML}{AB2346}
\definecolor{darkyellow}{HTML}{FFBA08}
\definecolor{tensor}{rgb}{0.5,0.8,0.5}
\definecolor{isometry}{rgb}{0.8,0.8,1}
\definecolor{unitary}{rgb}{0.8,0.5,.5}
\definecolor{gate}{rgb}{1.0,1.0,1.0}

\newcommand{\simplematrix}[3]{
	\begin{scope}[shift={(#1)}]
		\draw (-1,0) -- (1,0);
		\filldraw[fill=#3] (-1/2,-1/2) -- (-1/2,1/2) -- (1/2,1/2) -- (1/2,-1/2) -- (-1/2,-1/2);
		\draw (0,0) node {\scriptsize #2};
	\end{scope}
}

\newcommand{\simplematrixvertical}[3]{
	\begin{scope}[shift={(#1)}]
		\draw (0,-1) -- (0,1);
		\filldraw[fill=#3] (-1/2,-1/2) -- (-1/2,1/2) -- (1/2,1/2) -- (1/2,-1/2) -- (-1/2,-1/2);
		\draw (0,0) node {\scriptsize #2};
	\end{scope}
}

\newcommand{\rectangularsimplematrix}[4]{
	\begin{scope}[shift={(#1)}]
		\draw (-1-#4,0) -- (1+#4,0);
		\filldraw[fill=#3] (-1/2-#4,-1/2) -- (-1/2-#4,1/2) -- (1/2+#4,1/2) -- (1/2+#4,-1/2) -- (-1/2-#4,-1/2);
		\draw (0,0) node {\scriptsize #2};
	\end{scope}
}

\newcommand{\MPSTensor}[3]{
	\begin{scope}[shift={(#1)}]
		\draw (-1,0) -- (1,0);
		\draw (0,1) -- (0,0);
		\filldraw[fill=#3] (-1/2,-1/2) -- (-1/2,1/2) -- (1/2,1/2) -- (1/2,-1/2) -- (-1/2,-1/2);
		\draw (0,0) node {\scriptsize #2};
	\end{scope}
}

\newcommand{\MPSTensorLeft}[3]{
	\begin{scope}[shift={(#1)}]
		\draw (0,0) -- (1,0);
		\draw (0,1) -- (0,0);
		\filldraw[fill=#3] (-1/2,-1/2) -- (-1/2,1/2) -- (1/2,1/2) -- (1/2,-1/2) -- (-1/2,-1/2);
		\draw (0,0) node {\scriptsize #2};
	\end{scope}
}

\newcommand{\MPSTensorRight}[3]{
	\begin{scope}[shift={(#1)}]
		\draw (-1,0) -- (0,0);
		\draw (0,1) -- (0,0);
		\filldraw[fill=#3] (-1/2,-1/2) -- (-1/2,1/2) -- (1/2,1/2) -- (1/2,-1/2) -- (-1/2,-1/2);
		\draw (0,0) node {\scriptsize #2};
	\end{scope}
}

\newcommand{\MPSTensorrect}[4]{
	\begin{scope}[shift={(#1)}]
		\draw (-1,0) -- (1,0);
		\draw (0,1*#4) -- (0,0);
		\filldraw[fill=#3] (-1/2,-1/2*#4) -- (-1/2,1/2*#4) -- (1/2,1/2*#4) -- (1/2,-1/2*#4) -- (-1/2,-1/2*#4);
		\draw (0,0) node {\scriptsize #2};
	\end{scope}
}

\newcommand{\MPSTensorrecthoriz}[4]{
	\begin{scope}[shift={(#1)}]
		\draw (-1,0) -- (1,0);
		\draw (0,1) -- (0,0.5);
		\filldraw[fill=#3] (-1/2*#4,-1/2) -- (-1/2*#4,1/2) -- (1/2*#4,1/2) -- (1/2*#4,-1/2) -- (-1/2*#4,-1/2);
		\draw (0,0) node {\scriptsize #2};
	\end{scope}
}

\newcommand{\MPSTensordown}[3]{
	\begin{scope}[shift={(#1)}]
		\draw (-1,0) -- (1,0);
		\draw (0,-1) -- (0,0);
		\filldraw[fill=#3] (-1/2,-1/2) -- (-1/2,1/2) -- (1/2,1/2) -- (1/2,-1/2) -- (-1/2,-1/2);
		\draw (0,0) node {\scriptsize #2};
    \end{scope}
}
\newcommand{\MPSTensordownrecthoriz}[4]{
	\begin{scope}[shift={(#1)}]
		\draw (-1,0) -- (1,0);
		\draw (0,-1) -- (0,0);
		\filldraw[fill=#3] (-1/2*#4,-1/2) -- (-1/2*#4,1/2) -- (1/2*#4,1/2) -- (1/2*#4,-1/2) -- (-1/2*#4,-1/2);
		\draw (0,0) node {\scriptsize #2};
	\end{scope}
}

\newcommand{\InverseMPS}[3]{
	\begin{scope}[shift={(#1)}]
		\draw (-0.3,0) -- (-0.3,1);
        \draw (0.3,0) -- (0.3,1);
		\draw (0,-1) -- (0,0);
		\filldraw[fill=#3] (-1/2,-1/2) -- (-1/2,1/2) -- (1/2,1/2) -- (1/2,-1/2) -- (-1/2,-1/2);
		\draw (0,0) node {\scriptsize #2};
	\end{scope}
}

\newcommand{\ETensor}[5]{
	\begin{scope}[shift={(#1)}]
		\draw (-1-#4 +0.5,0) -- (1+ #4 -0.5,0);
        \draw (-1-#4 +0.5,#5) -- (1+#4 -0.5,#5);
		\filldraw[fill=#3] (-#4,-0.2) -- (-#4,#5 +0.2) -- (#4,#5 +0.2) -- (#4,-0.2) -- (-#4,-0.2);
		\draw (0,#5 /2) node {\scriptsize #2};
	\end{scope}
}

\newcommand{\RightVector}[5]{
	\begin{scope}[shift={(#1)}]
		\draw (-1-#4 +0.5,0) -- (0,0);
        \draw (-1-#4 +0.5,#5) -- (0,#5);
		\filldraw[fill=#3] (-#4,-0.2) -- (-#4,#5 +0.2) -- (#4,#5 +0.2) -- (#4,-0.2) -- (-#4,-0.2);
		\draw (0,#5 /2) node {\scriptsize #2};
	\end{scope}
}

\newcommand{\PhysicalOperator}[3]{
    \begin{scope}[shift={(#1)}]
		\draw (0,-1) -- (0,1);
		\filldraw[fill=#3] (-1/2,-1/2) -- (-1/2,1/2) -- (1/2,1/2) -- (1/2,-1/2) -- (-1/2,-1/2);
		\draw (0,0) node {\scriptsize #2};
	\end{scope}
}

\newcommand{\FullMPS}[3]{
	\begin{scope}[shift={(#1)}]
		\draw[shift={(0,0)},dotted] (0,0) -- (4.5,0);
        \MPSTensor{0,0}{#2}{#3}
        \MPSTensor{1.5,0}{#2}{#3}
        \MPSTensor{4.5,0}{#2}{#3}
        \draw (-1,0) -- (-1,-0.8) -- (5.5,-0.8) -- (5.5,0);
	\end{scope}
}

\newcommand{\FullMPSX}[5]{
	\begin{scope}[shift={(#1)}]
		\draw[shift={(0,0)},dotted] (0,0) -- (4.5,0);
        \MPSTensor{0,0}{#2}{#4}
        \MPSTensor{1.5,0}{#2}{#4}
        \MPSTensor{4.5,0}{#2}{#4}
        \simplematrix{-1.5,0}{#3}{#5}
        \draw (-2.5,0) -- (-2.5,-0.8) -- (5.5,-0.8) -- (5.5,0);
	\end{scope}
}

\makeatletter
\newcommand*{\balancecolsandclearpage}{%
  \close@column@grid
  \cleardoublepage
  \twocolumngrid
}
\makeatother

\begin{document}

\title{The product structure of MPS-under-permutations}

\author{Marta Florido-Llin\`as}
%\email{marta.florido.llinas@mpq.mpg.de}
\affiliation{
Max-Planck-Institut f{\"{u}}r Quantenoptik, Hans-Kopfermann-Str. 1, 85748 Garching, Germany
}
\affiliation{
Munich Center for Quantum Science and Technology (MCQST), Schellingstr. 4, 80799 M{\"{u}}nchen, Germany
}

\author{\'Alvaro M. Alhambra}
%\email{alvaro.alhambra@csic.es}
\affiliation{
Instituto de F\'isica Te\'orica UAM/CSIC, C/ Nicol\'as Cabrera 13-15, Cantoblanco, 28049 Madrid, Spain}

\author{Rahul Trivedi}
%\email{rahul.trivedi@mpq.mpg.de}
\affiliation{
Max-Planck-Institut f{\"{u}}r Quantenoptik, Hans-Kopfermann-Str. 1, 85748 Garching, Germany
}
\affiliation{
Munich Center for Quantum Science and Technology (MCQST), Schellingstr. 4, 80799 M{\"{u}}nchen, Germany
}

\author{Norbert Schuch}
%\email{norbert.schuch@univie.ac.at}
\affiliation{University of Vienna, Faculty of Mathematics, Oskar-Morgenstern-Platz 1, 1090 Wien, Austria}
\affiliation{University of Vienna, Faculty of Physics, Boltzmanngasse 5, 1090 Wien, Austria}

\author{David P\'erez-Garc\'ia}
%\email{dperezga@ucm.es}
\affiliation{Departamento de An\'alisis Matem\'atico, Universidad Complutense de Madrid, 28040 Madrid, Spain}

\author{J. Ignacio Cirac}
%\email{ignacio.cirac@mpq.mpg.de}
\affiliation{
Max-Planck-Institut f{\"{u}}r Quantenoptik, Hans-Kopfermann-Str. 1, 85748 Garching, Germany
}
\affiliation{
Munich Center for Quantum Science and Technology (MCQST), Schellingstr. 4, 80799 M{\"{u}}nchen, Germany
}%

\date{\today}

\begin{abstract}
Tensor network methods have proved to be highly effective in addressing a wide variety of physical scenarios, including those lacking an intrinsic one-dimensional geometry. In such contexts, it is possible for the problem to exhibit a weak form of permutational symmetry, in the sense that entanglement behaves similarly across any arbitrary bipartition. In this paper, we show that translationally-invariant (TI) matrix product states (MPS) with this property are trivial, meaning that they are either product states or superpositions of a few of them. The results also apply to non-TI generic MPS, as well as further relevant examples of MPS including the W state and the Dicke states in an approximate sense. Our findings motivate the usage of ans\"atze simpler than tensor networks in systems whose structure is invariant under permutations. %Additionally, they shed light on the fundamental challenge of understanding how the presence of certain types of symmetry restrictions can heavily constrain the entanglement structure and generality power of the class of states under consideration.
\end{abstract}

\maketitle

\section{Introduction} \label{sec:introduction}

\subsection{Motivation}

Tensor network (TN) methods, dating back to the inception of the DMRG algorithm \cite{DMRG92}, are one of our best tools to study quantum systems of many particles. The key property that makes them so useful is the fact that the correlation structure of many quantum states and operators is well captured by the geometry of the network. This means that these methods often yield accurate answers to problems involving low energy physics \cite{Verstraete2006,Schollwock_2011}, short-time dynamics \cite{Vidal_2004,Osborne_2006} or thermal equilibrium \cite{Verstraete_2004,White_2009,Kuwahara_2020}, among others. 

This is particularly the case for one-dimensional (1D) systems, for which the matrix product state (MPS) ansatz has been extensively studied \cite{fannes_finitely_1992, Perez-Garcia2007}, as well as for 2D systems \cite{verstraete_2004_PEPS}, and other geometries such as trees \cite{Murg2015}. Even further, the scope of tensor networks these days goes beyond physics, including other fields like machine learning \cite{stoudenmire2017supervised,novikov2017exponential,glasser2019expressive}, theoretical computer science \cite{Biamonte_2015, Kourtis_2019, Liu_2023}, and numerical methods for solving PDEs \cite{Khoromskij_2012, GarciaRipoll_2021, GarciaMolina_2023, Richter_PDEs_2021}.

In order to establish conditions under which MPS accurately approximate 1D chains with an efficient scaling of the bond dimension $D$, the scaling of the entanglement entropy of contiguous blocks can be used. For instance, states whose block Rényi entropies $S_\alpha$ for $\alpha < 1$ are upper bounded by a constant admit efficient MPS representations that describe them exactly if $\alpha = 0$, or approximately if $0 < \alpha < 1$ \cite{schuch2008entropy}. In such situations, expectation values can be computed with a time scaling as $O(D^3)$. 

This idea illustrates how certain physical constraints on states and models determine the most efficient ans\"atze to represent them. Other examples are states exhibiting permutational invariance, which are well approximated by convex combinations of identical product states (as shown by the so-called quantum de Finetti theorems \cite{caves_deFinetti_2002, renner_symmetry_2007, christandl2007definetti, brandao2013definetti}), or ground states of highly connected Hamiltonians, for which mean-field ansätze consisting of product states suffice \cite{bansal2007classical, gharibian2012approximation, brandao_product-state_2016}.

When using the MPS framework to study a specific system, a first necessary step is to choose a suitable one-dimensional ordering of the degrees of freedom. Quantum many-body systems, in particular quantum spin chains, naturally possess such an ordering. Yet, the MPS architecture (also known as tensor train \cite{oseledets_2011_tensor-trains}) is being used in contexts such as machine learning and numerical modelling, which lack an intrinsic 1D geometry. In those cases, the challenge arises: How can we determine the most suitable ordering of the variables for an MPS description out of all possible permutations~\cite{vinyals2015order, acharya2022qubit, Li_permutation_2022}? Similarly, applying MPS to problems in quantum chemistry \cite{Baiardi_chemistry_DMRG_2020, Chan_highly_2002, Chan_chemistry_DMRG_2011}, where the underlying interaction graph is usually highly connected, first requires finding a suitable 1D arrangement of the orbitals. While this ``ordering problem'' can be addressed through educated guesses \cite{legeza_optimizing_2003, rissler_measuring_2006, barcza2011} and specific algorithms \cite{krumnow2016, Moritz_2004, acharya2022qubit, Li_permutation_2022}, it can happen that they yield equally good accuracies regardless of the arrangement of the particles. This is the type of scenario that we study in this work. 

\subsection{Results}

In this work, we consider quantum many-body states, or more generally vectors in a tensor product space $\ket\psi\in(\mathbb C^d)^{\otimes N}$, with the property that they are well described by Matrix Product States (i.e., tensor trains), regardless of the ordering chosen. Here, ``well described'' can refer to either an exact MPS description or an approximate description with a given bond dimension $D$. Equivalently, this means that for any partition of the system into two parts, the Schmidt rank (i.e., the number of non-zero singular values, or the $0$-R\'enyi entropy) is either bounded by a constant $D$, or the state is well approximated by a state with such a bounded rank. We term this class of states (exact or approximate) \emph{MPS-under-permutations}.\footnote{Note that these states need by no means be permutationally invariant themselves; rather, they only need to possess a small amount of entanglement under any permutation of their sites.}

We then prove that, under some broadly applicable conditions, the class of exact (or approximate) MPS-under-permutations can, in fact, be written exactly (or approximately) as a sum over a small number of product states, that is, 
\begin{equation}
\label{eq:intro-central-result}
    \ket\psi \approx \sum_{i=1}^b \ket{\phi_i^1}\otimes \ket{\phi_i^2}
            \otimes\ket{\phi_i^3}\otimes\dots\ ,
\end{equation}
where the equality holds exactly (approximately) for exact (approximate) MPS-under-permutations. Expressed in quantum information terminology, $\ket\psi$ has a GHZ-type or cat-state-like form. This implies that MPS-under-permutations form a very special and restrictive subclass of MPS which does, in fact, have a much simpler representation than a general MPS, and which is significantly more efficient both in terms of the number of parameters as well as for computational purposes.

More specifically, what we show is the following: If $\ket\psi$ is an exact (approximate) MPS-under-permutations which is translationally invariant in some ordering, and its MPS description in the corresponding ordering has $b$ blocks in its canonical form,\footnote{Namely, the conventional canonical form for MPS
(cf.~Ref.~\cite{Cirac_review_2021}) which we will introduce later.} then
$\ket\psi$ is exactly (approximately) of the form
\eqref{eq:intro-central-result}, where each term in the sum is 
a tensor product $\ket{\phi_i}^{\otimes N}$ of identical states $\ket{\phi_i}$ on all $N$ components of $(\mathbb C^d)^{\otimes N}$.
If blocking is required in order to obtain the canonical form, the tensor product structure only shows up on the level of the blocks. In particular, if for the translationally invariant MPS representation of $\ket\psi$, $b=1$---which is satisfied by a generic MPS, and is equivalent to the absence of long-range correlations in the state\footnote{On a technical level, this means that the MPS is \emph{injective}~\cite{Cirac_review_2021}, see later}---then $\ket\psi$ is exactly (approximately) of tensor product form. 

Finally, if we drop the condition of translational invariance and consider exact or approximate MPS-under-permutations which have an injective MPS representation
(i.e.\ with $b = 1$), we show that $\ket\psi$ is still exactly or approximately
of tensor product form, but the tensor product no longer factorizes over all $N$
components in $(\mathbb C^d)^{\otimes N}$ since it can contain some small
clusters whose size depends on the amount of correlations.

\subsection{Discussion}

How surprising is this result, and do we really need to impose the aforementioned assumptions on the states? To understand these questions better, let us consider the special case of permutationally invariant states $\ket\psi\in(\mathbb C^d)^{\otimes N}$. It is well-known that permutationally invariant states can be expressed as superpositions of product states as $\ket{\psi} = \sum_{i=0}^N c_i \ket{\phi_i}^{\otimes N}$ ~\cite{Watrous_2018}. If we now have that $\ket\psi$ is an exact MPS-under-permutations, that is, it has a low Schmidt rank across every bipartition, it seems natural that we should be able to restrict the number of terms in the sum, possibly breaking the permutational symmetry, i.e., to write  $\ket{\psi} = \sum_{i=1}^r c_{i} \bigotimes_j \ket{\phi_{ij}}$, where $r$ depends on the Schmidt rank rather than $N$.

Yet, this is not the case. An illustrative example is the W state \cite{dur2000three}, $\ket{W_N} = \ket{10\dots 0} + \ket{01\dots 0} + \dots + \ket{00 \dots 1}$, a paradigmatic state in the context of quantum information. It is permutationally invariant and admits an MPS representation of bond dimension 2, that is, it has the exact MPS-under-permutation property. Still, it cannot be expressed as a sum of fewer than $N$ product states \cite{landsberg2011tensors}. At the same time, however, it can be approximated arbitrarily well by a sum of just two product states---specifically, $\ket{W_N} = \lim_{\varepsilon \to 0} \frac{1}{2\varepsilon} ([\ket{0} + \varepsilon \ket{1}]^{\otimes N} - [\ket{0} - \varepsilon \ket{1}]^{\otimes N})$, as the leading order cancels. These two notions of the number of terms needed in the sum are known as \emph{tensor rank} (for the exact description) and \emph{border rank} (for a description in the limit of $\varepsilon\to 0$), and are two key quantities in the field of algebraic complexity theory. For the W state, these values are known to be $N$ and $2$, respectively \cite{landsberg2011tensors}. However, finding them for a general state has been proven to be  NP-hard \cite{hillar_2013_tensor-NP-hard, landsberg2011tensors}: This illustrates that we cannot expect a solution to the question as to what the rank $b$ in \eqref{eq:intro-central-result} is for exact and approximate MPS-under-permutations in full generality, and thus explains the necessity of additional conditions on the families of states, such as those featuring in our results.

Our findings can be potentially understood as an approximate version of the quantum de Finetti theorem \cite{caves_deFinetti_2002, renner_symmetry_2007, christandl2007definetti, brandao2013definetti} in the context of MPS. As an application, we show that whenever TN methods yield accurate solutions irrespective of the particle ordering in certain contexts, such as in finding the unique ground state of gapped Hamiltonians, it is because the solution contains little entanglement in the first place. This justifies the use of a much simpler ansatz consisting of product states or superpositions of a few of them, for some problems where we expect entanglement to behave similarly regardless of the ordering. Equivalently, if MPS methods outperform mean-field approaches in particular scenarios without an in-built 1D geometry, there must exist at least a particle arrangement that is significantly better than others. 

One might reasonably ask: aren’t MPS already efficient enough? Even though most MPS operations scale polynomially in the bond dimension $D$ (e.g. computing expectation values costs $O(ND^3)$ \cite{orus_2014_introTNs}), the repeated sweeps of DMRG or time-evolution algorithms can still prove very costly in practice, even at moderate $D$ \cite{Schollwock_2011}. Therefore, when the full expressiveness of an MPS is not required, using instead a simpler variational family can yield significant savings, just as a superposition of $k$ product states reduces the cost of expectation-value computations to $O(Nk^2)$.

%\alv{should the angle here not be: ok, MPS are a simple piece of code, but you are still over-doing it with them - maybe we could suggest all-to-all hamiltonians, although not sure there are specific useful examples} \textcolor{blue}{$\to$ I've tried to condense the paragraph and express this point a bit better. Is it ok?}

This product-type ansatz corresponds to the well-established \textit{canonical polyadic decomposition} (CPD) in the broader tensor community. For decades, CPD has been a standard tool in fields ranging from psychometrics to chemistry and signal processing, due to its compactness, interpretability and, under mild conditions, uniqueness \cite{harshman_1970_CP-decomposition, landsberg2011tensors, Kolda_2009_tensor-decompositions-review, Stegeman_2007_uniqueness-cpd}. Although extracting the CPD is generally more challenging than building the tensor-train/MPS, the CPD remains central to tensor methods, with active research into more efficient algorithms to find it, including approaches that convert an existing MPS into CPD form \cite{Yu_2025_comparison-CP-algorithms, Prevost_2025_CPD-from-TT}. This enduring interest further highlights the practical value of the product-type ansätze that we consider in this work.

\section{Preliminaries} \label{sec:setting}

\subsection{MPS-under-permutations definition}

An MPS on $N$ particles can be expressed as
\begin{align}
    \ket{\psi} &:= \sum_{i_1 \dots i_N=1}^d \Tr[A^{[1],i_1} A^{[2],i_2} \dots A^{[N],i_N}] \ket{i_1 \dots i_N} \nonumber \\
    &= \begin{tikzpicture}[scale=.55, baseline={([yshift=-1ex]current bounding box.center)}, thick]
            \begin{scope}[shift={(0,0)}]
        		\draw[shift={(0,0)},dotted] (0.5,0) -- (4,0);
                \MPSTensor{(0,0)}{$A^{[1]}$}{purple}
                \MPSTensor{1.5,0}{$A^{[2]}$}{purple}
                \MPSTensor{4.5,0}{$A^{[N]}$}{purple}
                \draw (-1,0) -- (-1,-0.8) -- (5.5,-0.8) -- (5.5,0);
        	\end{scope}
        \end{tikzpicture}
    \in (\mathbb{C}^d)^{\otimes N},
    \label{eq:defMPS}
\end{align}
where $A^{[n],i}$ are $D_n \times D_{n+1}$ matrices ($D_{N+1} = D_1$) for all $i \in \{0, \dots, d-1\}$, $d$ being the local physical dimension. The quantity $D := \max_i D_i$ is referred to as the \textit{bond dimension} (for a detailed introduction to TNs, see e.g.\ \cite{Cirac_review_2021}). When all the tensors in Eq. (\ref{eq:defMPS}) are the same, the MPS is said to be \textit{translationally invariant} (TI) and denoted by $\ket{\psi_N(A)}$.

In this work, we study states with the property that the Schmidt rank across any arbitrary bipartition is upper bounded by a constant. We refer to such states as \textit{MPS-under-permutations} (MPS-up). 

\begin{mydef}{1}[MPS-up]
    \textit{A state $\ket{\psi}$ is an MPS-up$_{\varepsilon, D}$ if, for each bipartition of its $N$ particles, there is a state that is $\varepsilon$-close to $\ket{\psi}$ and whose Schmidt rank across the cut is upper bounded by $D$.}
\end{mydef}

Due to the fact that the MPS-up$_{0,D}$ property entails that $\max_m \text{rank}(\rho_m) \leq D$ where $\rho_m = \Tr_{m+1, \dots, N} \dyad{\psi}$, it implies that $\ket{\psi}$ is exactly an MPS of bond dimension $D$. Therefore, we can equivalently express the MPS-up$_{\varepsilon, D}$ property using the MPS language as follows.

\begin{mydef}{1}[MPS-up, MPS version] \label{def:MPS-up} 
    \textit{A state is an MPS-up$_{\varepsilon,D}$ if, for each permutation $\pi$ of its $N$ particles, there exists an MPS state with tensors $A^{[n],i}_\pi$ of bond dimensions $D_\pi^{[n]} \leq D$ that is at least $\varepsilon$-close to it,}
    \begin{equation*}
    \left\|
    \begin{array}{c}
		\begin{tikzpicture}[scale=.45, baseline={([yshift=-3ex]current bounding box.center)}, thick]
            % Draw the state:
            \filldraw[fill=purple] (0,0) -- (0,1) -- (2,1) -- (2,0) -- (0,0);
            \draw (1,0.5) node {\scriptsize $\ket{\psi}$};
            \draw (0.25,1) -- (0.25,1.5);
            \draw (0.5,1) -- (0.5,1.5);
            \draw[dotted] (0.75,1.25) -- (1.25,1.25);
            \draw (1.5,1) -- (1.5,1.5);
            \draw (1.75,1) -- (1.75,1.5);
            
            % Draw permutation on top:
            \filldraw[fill=yellow] (0,1.5) -- (0,2.5) -- (2,2.5) -- (2,1.5) -- (0,1.5);
            \draw (1,2) node {\scriptsize $U_\pi$};
            \draw (0.25,2.5) -- (0.25,3);
            \draw (0.5,2.5) -- (0.5,3);
            \draw[dotted] (0.75,2.75) -- (1.25,2.75);
            \draw (1.5,2.5) -- (1.5,3);
            \draw (1.75,2.5) -- (1.75,3);
		\end{tikzpicture}
	\end{array}
    - \begin{array}{c}
		\begin{tikzpicture}[scale=.55, baseline={([yshift=-3ex]current bounding box.center)}, thick]
            \begin{scope}[shift={(0,0)}]
        		\draw[shift={(0,0)},dotted] (0.5,0) -- (4,0);
                \MPSTensorLeft{(0,0)}{$A_\pi^{[1]}$}{amaranth}
                \MPSTensor{1.5,0}{$A_\pi^{[2]}$}{amaranth}
                \MPSTensorRight{4.5,0}{$A_\pi^{[N]}$}{amaranth}
                %\draw (-1,0) -- (-1,-0.8) -- (5.5,-0.8) -- (5.5,0);
        	\end{scope}
        \end{tikzpicture}
    \end{array} \right\| \leq \varepsilon,
    \end{equation*}
    \textit{where $U_\pi$ is the unitary operator that permutes the $N$ local Hilbert spaces according to $\pi$. We say that $\ket{\psi}$ is an exact MPS-up$_D$ if it is an MPS-up$_{0,D}$.}
\end{mydef}

Note that exact MPS-up are not necessarily permutationally invariant, since the MPS that approximates each permutation $\pi$ can be different. To tackle the problem of characterizing these states, we will begin our analysis with TI MPS, for which theoretical tools have been extensively developed in the literature \cite{Perez-Garcia2007, Cirac_review_2021}. 

 \subsection{Tensor networks background} \label{sec:TNs-background}

A TI MPS with tensor $A$ is called \textit{normal} if its matrices have no non-trivial common invariant subspace, and the associated CP map $\mathcal{E}(X) = \sum_{i=0}^{d-1} A^i X (A^i)^\dagger$ has a unique largest eigenvalue of magnitude (and value) equal to one \cite{Cirac2017}. These MPS display short-range correlations. 

In the general (or \textit{non-normal}) setting, any TI MPS with tensor $A$ can be expressed in terms of a so-called \textit{basis of normal tensors} (BNT) \(\{A_1, \dots, A_b\}\) as \cite{Cirac2017}
\begin{equation}    \label{eq:nonnormal_MPS_intro}  
    \ket{\psi_N(A)} \propto \sum_{j=1}^b \alpha_j \ket{\psi_N(A_j)},
\end{equation} 
where ($i$) each \(A_j\) is normal if every \(p\) sites have already been blocked to get rid of periodic subspaces, where \textit{blocking} refers to grouping together every $p$ physical sites into a new tensor $\tilde{A}$ with physical dimension $d^p$, such that $\tilde{A}^{i_1 i_2 \dots i_p} := A^{i_1} A^{i_2} \dots A^{i_p}$, and $(ii)$ there is \textit{block-injectivity}, meaning that each element $A_j$ can be accessed separately from the others just by acting on the physical index if at least every $L_{BI}$ sites are blocked together, where $L_{BI}$ is referred to as the \textit{block-injectivity length} and is always upper bounded by a constant independent of the system size.

A normal tensor $A$ is \textit{injective} if the map 
\begin{align*}
    \Gamma: \mathcal{M}_D(\mathbb{C}) &\to \mathbb{C}^d \\
    \begin{array}{c} \begin{tikzpicture}[scale=.45,thick,baseline={([yshift=0ex]current bounding box.center)}]
        \simplematrix{0,0}{$X$}{yellow}
    \end{tikzpicture} \end{array}  
    &\mapsto
    \begin{array}{c} 
		\begin{tikzpicture}[scale=.45,thick,baseline={([yshift=0ex]current bounding box.center)}]
			\simplematrix{0,0}{$X$}{yellow}
            \MPSTensor{1.5,0}{$A$}{purple}
            \draw (-1,0) -- (-1,-0.8) -- (2.5,-0.8) -- (2.5,0);
		\end{tikzpicture}
    \end{array}
\end{align*}
is injective, or equivalently, if there exists a left-inverse $A^{-1} : \mathbb{C}^d \to \mathcal{M}_D(\mathbb{C})$ such that
\begin{equation} \label{eq:left-inverse}
    \begin{array}{c} 
		    \scalebox{0.85}{
            \begin{tikzpicture}[scale=.5,thick,baseline={([yshift=0ex]current bounding box.center)}]
            \MPSTensor{0,0}{$A$}{purple}
            \MPSTensordown{0,1.5}{$A^{-1}$}{purple}
		\end{tikzpicture}}
    \end{array} =
    \begin{array}{c} 
		\scalebox{0.85}{\begin{tikzpicture}[scale=.5,thick,baseline={([yshift=0ex]current bounding box.center)}]
            \draw (-1,0) -- (-0.15,0) -- (-0.15,1.5) -- (-1,1.5);
            \draw (1,0) -- (0.15,0) -- (0.15,1.5) -- (1,1.5);
		\end{tikzpicture}}
    \end{array}.
\end{equation}
After blocking a constant number of physical sites together in the TI setting, any normal tensor $A$ becomes injective \cite{Sanz2010, Michalek2018}. This number is called the \textit{injectivity length} $L_I$ of $A$, and it is upper bounded by $L_I \leq 2D^2(6+\log_2 D)$. More generally, after blocking every $L_{BI}$ sites, any non-normal tensor $A$ becomes a direct sum of injective tensors that can be separately accessed through the physical index.

An object of special interest for MPS that will appear in some of the proofs is the so-called \textit{transfer matrix}, defined as
\begin{equation*}
	\begin{array}{c} \begin{tikzpicture}[scale=.5,thick,baseline={([yshift=1ex]current bounding box.center)}]
        \ETensor{0,0}{$\mathbb{E}$}{purple}{0.5}{1.5}
    \end{tikzpicture} \end{array}  
    := \sum_{i=1}^d A^i \otimes (A^{i})^* =
    \begin{array}{c}
		\begin{tikzpicture}[scale=.45,thick,baseline={([yshift=1ex]current bounding box.center)}]
			\MPSTensor{0,0}{$A^\ast$}{purple}
			\MPSTensordown{0,1.5}{$A$}{purple}
		\end{tikzpicture}
    \end{array},
\end{equation*}
where $A^\ast$ represents component-wise complex conjugation. By exploiting the gauge freedom 
$A^i\leftrightarrow XA^i X^{-1}$
of the TI MPS and normalizing it, we can always write a normal tensor in \textit{left canonical form} \cite{Schollwock_2011}, so that $\mathbb{E}^m \xrightarrow{\scriptsize m \to\infty} \dyad{\Lambda}{\mathds{1}}$, where $\ket{\mathds{1}}$ and $\ket{\Lambda}$ are the vectorised versions of the identity matrix, and of a diagonal positive matrix $\Lambda$, respectively. Note that $\| \ket{\psi_N(A)} \|^2  = \Tr[\mathbb{E}^N] \to 1$ as $N \to \infty$.

\section{Summary of results} \label{sec:summary_results}

\begin{table}[!h]
    \centering
    \begin{tabularx}{0.425\textwidth}{c||c|c}
         & \makecell{exact MPS-up$_{D}$ \\ {\footnotesize (fixed $N > N_0$)}} 
         & \makecell{Approx. MPS-up$_{\varepsilon,D}$ \\ {\footnotesize($N \to \infty$)}} \\ \hline\hline
        \makecell{TI normal \\ MPS} 
        & \makecell{$=$ product \\ \scriptsize{(Prop. \ref{prop:rank_counting_normal})}} 
        & \makecell{$=$ product \\ \scriptsize{(if $\varepsilon < \dots $;} \scriptsize{Prop. \ref{prop:purity_normal})}} \\ \hline
        \makecell{TI non-\\ normal} 
        & \makecell{$=$ GHZ-like \\ \scriptsize{(Thm. \ref{thm:exact_MPS-up_thm})}} 
        & \makecell{$=$ GHZ-like \\ \scriptsize{(if $\varepsilon < \dots$; Thm. \ref{thm:approx_MPS-up_thm})}} \\ \hline
        \makecell{Generic \\ Non-TI} 
        & \makecell{$=$ ``product'' \\ \scriptsize{(Prop. \ref{prop:rank_counting_non-uniform_injective})}} 
        & \makecell{$=$ ``product'' \\\scriptsize{(if exp. decay; Prop. \ref{prop:purity_non-TI_injective})}} \\ \hline
        W-type 
        & \multicolumn{2}{c}{$\approx$ GHZ-like \scriptsize{(Table \ref{table:non-TI})}} \\ 
    \end{tabularx}
    \caption{Summary of results.} \label{table:results}
\end{table}

In this work, we first show that the MPS-up property can only hold on TI MPS if the elements of the BNT have bond dimension one, and thus the state is necessarily a superposition of as many product states as the number of elements in the BNT. Indeed, for the exact MPS-up$_{D}$ property, this holds whenever $N$ is large enough, as summarised in the theorem below (see section \ref{sec:TI_exact_MPS-up}).
\begin{mythm}{1}[Exact TI MPS-up, informal] \label{thm:exact_MPSup_informal}
    Let $\ket{\psi}$ be a TI MPS with the exact MPS-up$_{D}$ property on $N$ sites, with sufficiently large $N$. Then, $\ket{\psi}$ can be written as a superposition of as many product states as the number of elements in the BNT of the MPS tensor.
\end{mythm}

%\begin{restatable}{theorem}{exactMPSUPthm} \label{thm:exact_MPS-up_thm}
%    Let $\ket{\psi_N(A)}$ be a TI MPS with the exact MPS-up$_{D}$ property on $N$ sites, with $N > pL_{BI}(\log_2 D + 1)$. Then, $\ket{\psi} = \sum_{i=1}^b \beta_i \ket{\phi_i}^{\otimes N}$, where $b$ denotes the number of elements in the BNT of tensor $A$.
%\end{restatable}

In section \ref{sec:TI_approx_MPS-up} we generalize this to approximate TI MPS-up$_{\varepsilon, D}$ with $\varepsilon \geq 0$, we consider families of TI MPS and impose restrictions on $\varepsilon$. 

\begin{mythm}{2}[Approximate TI MPS-up, informal] \label{thm:approx_MPSup_informal}
    Let $\{\ket{\psi_N}\}_N$ be a family of TI MPS with the approximate MPS-up$_{\varepsilon, D}$ property for all $N$. Suppose the MPS tensor has $b$ elements in its basis of normal tensors. Then, $\ket{\psi_N}$ can be written as a superposition of $b$ product states, provided that $\varepsilon < \frac{1}{4D}$ if $b=1$, or $\varepsilon < \frac{C}{8D}$ if $b > 1$, where $C$ is a constant that depends on the relation between the basis elements of the tensor.
\end{mythm}

\iffalse\begin{restatable}{theorem}{approxMPSUPthm} \label{thm:approx_MPS-up_thm}
    Let $\{\ket{\psi_N(A)}\}_N$ be a family of TI MPS with the approximate MPS-up$_{\varepsilon_N, D_N}$ property for all $N$ larger than some $N_0$, where $D_N = O(\text{poly}(N))$. 
    Let $b$ be the number of elements in the BNT of tensor $A$.
    Then, if there exists a positive sequence $(g_N)$ with $g_N = \Omega(1/\text{poly}(N))$, such that either
    \begin{enumerate}[(a)]
        \item $b = 1$ and $0 \leq \varepsilon_N < \frac{1}{4D_N}-g_N$, or
        \item $b > 1$ and $0 \leq \varepsilon_N <  \left( \frac{1}{4D_N} - g_N \right) \frac{\min_i |\alpha_{i}|}{2(\sum_{i} |\alpha_{i}|^2)^{\frac{1}{2}}}$, 
    \end{enumerate}
    where $\alpha_i$ denote the coefficients weighting each element of the BNT of $A$ according to Eq. \eqref{eq:nonnormal_MPS_intro}, then $\ket{\psi_N(A)} = \sum_{i=1}^b \beta_i \ket{\phi_i}^{\otimes N}$.
\end{restatable}\fi

We extend our results to more general settings in Section \ref{sec:non-TI}, where we prove the following statement.
\begin{mythm}{3}[Non-TI MPS-up, informal] \label{cor:nonTI_MPSup_informal}
    A non-TI MPS with exponentially decaying correlations and the MPS-up$_{\varepsilon, D}$ property can be written as a product state of almost all sites or blocks of a few sites.
\end{mythm}
In section \ref{sec:implications_mps_approx_algorithms} we illustrate how these results can be used to justify the usage of simpler mean-field ansätze under certain conditions in MPS approximation algorithms. Finally, we discuss some relevant examples like the W or Dicke states that lie beyond the scope of our assumptions (section \ref{sec:almost-product}).

Table \ref{table:results} shows a summary of the results and their corresponding locations in the paper. We use the term ``GHZ-like'' to denote änsatze of the form $\sum_i \beta_i \ket{\phi_i}^{\otimes N}$ and $\sum_i \beta_i \otimes_{j=1}^N \ket{\phi_{ij}}$ in the TI and non-TI cases, respectively, due to their similarity with the paradigmatic and long-range ordered GHZ-state $|\text{GHZ}^d_N\rangle = \sum_{i=0}^{d-1} \ket{i}^{\otimes N}$ \cite{Greenberger_1989}.

\section{Exact translationally invariant MPS-up} \label{sec:TI_exact_MPS-up}

We start by considering TI MPS $\ket{\psi_N(A)}$ with normal tensor $A$ and
bond dimension $D_\mathds{1}$, on a fixed number of particles $N$. We assume
that the state is an exact MPS-up$_{D}$, meaning that the Schmidt rank is
bounded across any arbitrary bipartition.\footnote{ Note that, in all the proofs
developed in this work, it is not necessary to assume that the MPS-up property
holds for all permutations. As long as it holds for one permutation $\pi \in
\mathcal{S}_N$ that leads to a contradiction when $N$ is large enough, the
conclusions follow (for example, the permutation $\pi$ in Eq.
(\ref{eq:def_permutation_rank_counting}) for the exact MPS-up case as in
Proposition \ref{prop:rank_counting_normal}, or the permutation sending each
$k$-th particle to the beginning of the chain for the approximate MPS-up case as
in Proposition \ref{prop:purity_normal}). For simplicity, however, we consider
the generic definition.} We do not, however, require that the state is also translational invariant in any other ordering of sites. For the sake of simplicity, let us momentarily assume the more restrictive property of injectivity on the normal tensor $A$.

\begin{lemma}\label{lemma:rank_counting_injective}
    An injective TI MPS with the exact MPS-up property on an even number $N > \log_2 D$ of particles is a product state $\ket{\psi} = \ket{\phi}^{\otimes N}$.
\end{lemma}
\begin{proof}
    The proof is based on a rank counting argument after a certain permutation is performed. Let $\pi \in \mathcal{S}_{N}$ be defined as
    {\small \begin{equation} \label{eq:def_permutation_rank_counting}
        \begin{cases}
            \pi(2k-1) = k, &\text{for } k = 1, \dots, \lceil\frac{N}{2}\rceil \\
            \pi(2k) = \lceil\frac{N}{2}\rceil + k , &\text{for } k = 1, \dots, \lfloor\frac{N}{2}\rfloor
        \end{cases}.
    \end{equation}
    This permutation places the particles in odd positions at the left, and the ones in even positions at the right. For $N = 6$, this means that
    \begin{equation*} \begin{tikzpicture}[scale=.45, baseline={([yshift=-3ex]current bounding box.center)}, thick]
        \foreach \x in {1,2,3,4,5,6}{
            \node[below] at (\x-1,-0.1) {\scriptsize $\x$};
		}
        \draw[thick, fill=purple] (0,0) circle (0.25);
        \draw[thick, fill=darkpurple] (1,0) circle (0.25);
        \draw[thick, fill=purple] (2,0) circle (0.25);
        \draw[thick, fill=darkpurple] (3,0) circle (0.25);
        \draw[thick, fill=purple] (4,0) circle (0.25);
        \draw[thick, fill=darkpurple] (5,0) circle (0.25);
        \node at (6,0) {$\xrightarrow{\pi}$};
        \draw[thick, fill=purple] (7+0,0) circle (0.25);
        \node[below] at (7,-0.1) {\scriptsize 1};
        \draw[thick, fill=purple] (7+1,0) circle (0.25);
        \node[below] at (8,-0.1) {\scriptsize 3};
        \draw[thick, fill=purple] (7+2,0) circle (0.25);
        \node[below] at (9,-0.1) {\scriptsize 5};
        \draw[thick, fill=darkpurple] (7+3,0) circle (0.25);
        \node[below] at (10,-0.1) {\scriptsize 2};
        \draw[thick, fill=darkpurple] (7+4,0) circle (0.25);
        \node[below] at (11,-0.1) {\scriptsize 4};
        \draw[thick, fill=darkpurple] (7+5,0) circle (0.25);
        \node[below] at (12,-0.1) {\scriptsize 6};
    \end{tikzpicture} . \end{equation*}}

    Let $\mathcal{A}^{-1}:=(A^{-1})^{\otimes N}$. On the one hand,
    \begin{equation}
    U_\pi \mathcal{A}^{-1} \ket{\psi} = U_\pi \left( 
    \begin{array}{c} 
		\begin{tikzpicture}[scale=.45,thick,baseline={([yshift=0ex]current bounding box.center)}]
            \draw (-0.5,-0.5) -- (-0.5,0) -- (0-0.15,0) -- (0-0.15,0.6);
            \draw (0+0.15,0.6) -- (0+0.15,0) -- (1-0.15,0) -- (1-0.15,0.6);
            \draw (1+0.15,0.6) -- (1+0.15,0) -- (1+0.5,0);
            \draw[dotted] (1+0.5,0) -- (3-0.5,0);
            \draw (3-0.5,0) -- (3-0.15,0) -- (3-0.15,0.6);
            \draw (3+0.15,0.6) -- (3+0.15,0) -- (3+0.5,0) -- (3+0.5,-0.5) -- (-0.5,-0.5);
		\end{tikzpicture}
    \end{array}  \right) =
    \begin{array}{c} 
		\begin{tikzpicture}[scale=.45,thick,baseline={([yshift=0ex]current bounding box.center)}]
            \draw (-1.5,-0.5) -- (-1.5,0) -- (-1-0.15,0) -- (-1-0.15,1.8);
            \draw (-1+0.15,1.8) -- (-1+0.15,0) -- (-1+1-0.15,0) -- (-1+1-0.15,0.5);
            \draw (-0.5,0) -- (0-0.15,0) -- (0-0.15,0.6);
            \draw (0+0.15,0.4) -- (0+0.15,0) -- (1-0.15,0) -- (1-0.15,0.5);
            \draw (1+0.15,0.6) -- (1+0.15,0) -- (1+0.5,0);
            \draw[dotted] (1.5,0) -- (3,0);
            \draw (3.5-0.5,0) -- (3.5-0.15,0) -- (3.5-0.15,0.4);
            \draw (3.5+0.15,0.6) -- (3.5+0.15,0) -- (4,0) -- (4,-0.5) -- (-1.5,-0.5);
            %
            % New permuted lines:
            \draw[thick] (0-0.15,0.6) to[out=90, in=-90] (3-0.15,1.8);
            \draw[thick] (0+0.15,0.4) to[out=90, in=-90] (3+0.15,1.6)--(3+0.15,1.8);
            \draw[thick] (1-0.15,0.5) to[out=90, in=-90] (0-0.15,1.8)--(-0.15,1.8);
            \draw[thick] (1+0.15,0.6) to[out=90, in=-90] (0+0.15,1.8);
            \draw[thick] (3.5-0.15,0.4) to[out=90, in=-90] (1.3-0.15,1.6) -- (1.3-0.15,1.8);
            \draw[thick] (3.5+0.15,0.6) to[out=90, in=-90] (1.3+0.15,1.8);
            %
            % Bipartition:
            \draw[densely dashed,draw=red] (2,-1.6)  -- (2,2.1);
            \draw (2.8,-1.3) node {\scriptsize $D^N_{\mathds{1}}$};
		\end{tikzpicture}
    \end{array} \label{eq:rank_counting_normal_proof_aux}.
    \end{equation}
    The Schmidt rank of this state across the half-chain bipartition $\{1, \dots, \frac{N}{2} \} \cup \{\frac{N}{2} + 1, \dots, N\}$ of the resulting state can be computed exactly, since it consists just of Bell pairs. Its value is $D_\mathds{1}^{N}$.

    On the other hand, noting that $U_\pi$ and $\mathcal{A}^{-1}$ satisfy that\footnote{
    Throughout the paper, we make the abuse of notation of denoting with the same symbol, $U_\pi$, any unitary operator that permutes according to $\pi \in \mathcal{S}_N$ any $N$ local Hilbert spaces, each of which could have different local dimensions. For instance, in the equation $U_\pi \mathcal{A}^{-1} = A_{\pi}^{-1} U_\pi$, the first $U_\pi$ acts on $(\mathcal{M}_{D_{\mathds{1}}}(\mathbb{C}))^{\otimes N}$, while the second $U_\pi$ acts on $(\mathbb{C}^{d})^{\otimes N}$.} $U_\pi \mathcal{A}^{-1} = \mathcal{A}_\pi^{-1} U_\pi$, i.e. 
    \begin{equation}
    U_\pi \mathcal{A}^{-1} \ket{\psi} = \mathcal{A}_\pi^{-1} U_\pi \ket{\psi} = \begin{array}{c}
		\begin{tikzpicture}[scale=.5, baseline={([yshift=-3ex]current bounding box.center)}, thick, every node/.style={scale=0.9}]
            \begin{scope}[shift={(0,0)}]
        		\draw[shift={(0,0)},dotted] (0.5+0.3,0) -- (4+0.3,0);
                \MPSTensorLeft{(0,0)}{$A_\pi^{[1]}$}{amaranth}
                \MPSTensor{1.8,0}{$A_\pi^{[2]}$}{amaranth}
                \MPSTensorRight{4.5,0}{$A_\pi^{[\tilde{N}]}$}{amaranth}
                %\draw (-1,0) -- (-1,-0.8) -- (5.5,-0.8) -- (5.5,0);
        	\end{scope}
            \filldraw[shift={(0,1.5)},fill=purple] (-1/2-0.1,-1/2) -- (-1/2-0.1,1/2) -- (1/2+0.1,1/2) -- (1/2+0.1,-1/2) -- (-1/2-0.1,-1/2);
            \draw (-0.3, 2) -- (-0.3, 2.5);
            \draw(0.3, 2) -- (0.3, 2.5);
            \draw (0,1.5) node {\scriptsize $\mathcal{A}^{-1}$};
            %
            %\InverseMPS{(0,1.5)}{$\mathcal{A}_{\pi,1}^{-1,[1]}$}{purple}
            \filldraw[shift={(1.5+0.3,1.5)},fill=purple] (-1/2-0.1,-1/2) -- (-1/2-0.1,1/2) -- (1/2+0.1,1/2) -- (1/2+0.1,-1/2) -- (-1/2-0.1,-1/2);
            \draw (1.5, 2) -- (1.5, 2.5);
            \draw(2.1, 2) -- (2.1, 2.5);
            \draw (1.8,1.5) node {\scriptsize $\mathcal{A}^{-1}$};
            %
            %\InverseMPS{(1.5,1.5)}{$\mathcal{A}_{\pi,2}^{-1,[2]}$}{purple}
            \filldraw[shift={(4.5,1.5)},fill=purple] (-1/2-0.1,-1/2) -- (-1/2-0.1,1/2) -- (1/2+0.1,1/2) -- (1/2+0.1,-1/2) -- (-1/2-0.1,-1/2);
            \draw (4.2, 2) -- (4.2, 2.5);
            \draw(4.8, 2) -- (4.8, 2.5);
            \draw (4.5,1.5) node {\scriptsize $\mathcal{A}^{-1}$};
            \draw[densely dashed,draw=red] (3.15,-1.2)  -- (3.15,2.6);
            \draw (3.8,-0.9) node {\scriptsize $\leq D$};
        \end{tikzpicture}
    \end{array}, \label{eq:rank_proof_aux}
    \end{equation}
    where we used the MPS-up$_{0,D}$ property to substitute $U_\pi \ket{\psi}$ by an MPS of bond dimension $D_\pi \leq D$.
    
    To be consistent with the previous expression, it is necessary that $D_\mathds{1}^{N/2} \leq D$, which can only hold if \textit{(i)} $D_\mathds{1} > 1$ but $N$ is not too large with respect to $D$, so that the equation is not violated, or \textit{(ii)} $D_\mathds{1} = 1$ and thus $\ket{\psi}$ is a product state. Situation \textit{(i)} is impossible for any $D_\mathds{1} > 1$ if $2^{N} > D$, and the claim follows.
\end{proof}

Note that the resulting product state is of the form $\ket{\psi} = \ket{\phi}^{\otimes N}$, where $\ket{\phi}$ is site-independent, due to the fact that the state is a TI MPS with bond dimension 1 (i.e., with the same tensor at each site).

The injectivity assumption can be easily relaxed to normality by noting that any normal tensor becomes injective when blocking every $L_I$ sites, where $L_I$ is the injectivity length. This observation allows us to extend the desired statement to normal TI exact MPS-up. The proof, which is analogous to the one of Lemma \ref{lemma:rank_counting_injective}, is provided in Appendix \ref{app:proof_rankcountingnormal}.
\begin{restatable}{proposition}{rankcountingnormal} \label{prop:rank_counting_normal}
    A normal TI MPS with the exact MPS-up property on $N > L_I (\log_2 D + 1)$ particles is a product state $\ket{\psi} = \ket{\phi}^{\otimes N}$.
\end{restatable}

In the general case where tensor $A$ is not necessarily normal, we show that any TI exact MPS-up exhibits a GHZ-type entanglement structure if the number of particles $N$ is large enough, since it can be written as a superposition of as many linearly independent product states as the number of elements in the basis of normal tensors of $A$. The proof follows the same strategy as in the normal case, but using block-injectivity to isolate the individual sectors. It can be found in Appendix \ref{app:proof_rankcounting-nonnormal}.
\begin{mythm}{1}[Exact MPS-up] \label{thm:exact_MPS-up_thm}
    Let $\ket{\psi_N(A)}$ be a TI MPS with the exact MPS-up$_{D}$ property on $N$ sites, with $N > pL_{BI}(\log_2 D + 1)$. Then, $\ket{\psi} = \sum_{i=1}^b \beta_i \ket{\phi_i}^{\otimes N}$, where $b$ denotes the number of elements in the BNT of tensor $A$.
\end{mythm}
Note that if the tensor is not already in the canonical form leading to Eq. \eqref{eq:nonnormal_MPS_intro} (i.e., if one still needs to block a certain small amount of sites together, due to the presence of periodic subspaces), then the conclusion is not that the state is a sum of TI product states over all individual sites, but rather a sum of TI products over small clusters (an illustrative example is the state $\ket{10101010\dots} + \ket{01010101\dots}$, which is a sum of TI products of clusters of 2 sites).

\section{Approximate translationally invariant MPS-up}
\label{sec:TI_approx_MPS-up}

We now study the more general scenario of TI MPS admitting an approximate MPS representation for each possible ordering, but not necessarily an exact one, as in the previous section. Equivalently, the state is an MPS-up$_{\varepsilon, D}$ with $\varepsilon \geq 0$. 

Specifically, given a normal tensor $A$, we will consider the corresponding family of TI MPS $\{\ket{\psi_N(A)}\}_N$. In this situation, we cannot use the rank counting argument as we did for exact MPS-up in section \ref{sec:TI_exact_MPS-up}, since the Schmidt rank is not stable under perturbations. Instead, we use an alternative reasoning which relies on the fact that the MPS-up property imposes a lower bound on the purity for all bipartitions to get a contradiction, inspired by \cite{Rolandi2020}.

\begin{proposition}
\label{prop:purity_normal}
    Let $\{\ket{\psi_N(A)}\}_N$ be a family of TI MPS with normal tensor $A$. If it has the MPS-up$_{\varepsilon, D}$ property for all $N$ larger than some $N_0$, with $\varepsilon < \frac{1}{4D}$, then $A$ has bond dimension $1$ and thus $\ket{\psi_N(A)}$ is a product state, $\ket{\psi_N(A)} = \ket{\phi}^{\otimes N}$, for all $N$.
\end{proposition}

\begin{proof}
Given any $k \in \mathbb{N}$, let $S^N_k \subseteq \{1, \dots, N\}$ be the subset consisting of every $k$-th particle, and $\rho_S$ the reduced density matrix on $S^N_k$, as depicted in the top left part of the diagram below. 

Take $N = nk > N_0$ for some $n \in \mathbb{N}$. A lower bound on the purity of $\rho_S$ can be readily obtained, by considering the permutation $\pi$ sending all particles in $S^N_k$ to the beginning of the chain, as shown in the bottom left of the diagram.
\begin{equation} \label{eq:permutation_approx_case}
    \begin{array}{c} \begin{tikzpicture}[scale=.5, baseline={([yshift=-3ex]current bounding box.center)}, thick]
        \foreach \x in {0,1,3,4,5,6}{
            \draw[thick, fill=purple] (\x,0) circle (0.25);
            \draw[thick, fill=purple] (\x,2) circle (0.25);
		}
        \draw[thick, fill=darkpurple] (0,0) circle (0.25);
        \draw[thick, fill=darkpurple] (1,0) circle (0.25);
        \draw[thick, fill=darkpurple] (3,0) circle (0.25);
        \node[below] at (0,0) {\scriptsize $k$};
        \node[below] at (1,0) {\scriptsize $2k$};
        \node[below] at (3,0) {\scriptsize $nk$};
        \node[below] at (4,0) {\scriptsize $1$}; 
        \node[below] at (5,0) {\scriptsize $2$};
        \node at (2,0) {\small $\dots$};
        \node at (7.2,0) {\small $\dots$};
        \node at (0.8,0.9) {\scriptsize $\pi$};

        \draw[thick, fill=darkpurple] (2,2) circle (0.25);
        \draw[thick, fill=darkpurple] (5,2) circle (0.25);
        \node at (7.2,2) {\small $\dots$};
        \node[below] at (2,2) {\scriptsize $k$}; 
        \node[below] at (5,2) {\scriptsize $2k$}; 

        \draw[->, thick] (1.9,1.2) -- (0.1,0.35);
        \draw[->, thick] (4.9,1.2) -- (1.1,0.35);

        % Now we draw the right part with the MPS:
        \node at (8.5,2) {\scriptsize $\leftrightarrow$};
        \node at (8.5,0) {\scriptsize $\leftrightarrow$};

        % Draw the MPS:
        \FullMPS{10.75,2}{$A$}{purple}
        %\FullMPS{10.75,0}{$A_\pi$}{amaranth}
        \begin{scope}[shift={(10.75,0)}]
            \draw[shift={(0,0)},dotted] (0.5,0) -- (4,0);
            \MPSTensorLeft{(0,0)}{$A_\pi^{[1]}$}{amaranth}
            \MPSTensor{1.5,0}{$A_\pi^{[2]}$}{amaranth}
            \MPSTensorRight{4.5,0}{$A_\pi^{[N]}$}{amaranth}
            %\draw (-1,0) -- (-1,-0.8) -- (5.5,-0.8) -- (5.5,0);
        \end{scope}
        
        \draw[densely dashed,draw=red] (3.5,-1.1)  -- (3.5,0.7);
        \draw[densely dashed,draw=red] (13.75,-1.2)  -- (13.75,0.8);
        \draw[right] (13.75,-1) node {\scriptsize $\leq D$};
    \end{tikzpicture} \end{array}.
\end{equation}

The particles ordered according to $\pi$ admit an approximate MPS representation with tensor $A_\pi$ of bond dimension upper bounded by $D$, due to the MPS-up$_{\varepsilon, D}$ property. This means that the purity of the MPS approximation $A_\pi$ across the dashed red line, $\Tr[(\rho_{\pi,S})^2]$, is lower-bounded by $1/D$. Hence, the purity of $\rho_S$ satisfies
\begin{align} 
    \Tr[(\rho_S)^2] &\geq \text{Tr}[(\rho_{\pi,S})^2] - \left| \Tr[(\rho_S)^2] - \text{Tr}[(\rho_{\pi,S})^2] \right| \nonumber \\
    &\geq \frac{1}{D} - 4 \varepsilon \label{eq:purity1},
\end{align}
where Lemma \ref{lemma:trace_distance} of Appendix \ref{app:proof_technical-lemma} was used to relate the purity to the distance between the states. 

Now, we compute the purity $\Tr[(\rho_S)^2]$ starting from its MPS representation $A$ in left-canonical form and taking the limit of $k \to \infty$, which results in
\begin{align}
    \Tr[(\rho_{S})^2] &= \frac{1}{\Tr[\mathbb{E}^{nk}]^2} \Tr[ \begin{array}{c}
		\begin{tikzpicture}[scale=.45, baseline={([yshift=-5ex]current bounding box.center)}, thick]
            \ETensor{0,0}{$\mathbb{E}^{k-1}$}{purple}{0.8}{0.7}
            \ETensor{0,1.4}{$\mathbb{E}^{k-1}$}{purple}{0.8}{0.7}
            % Draw the swaps:
            \draw (1.3,0) -- (1.8,1.4);
            \draw (1.3,1.4) -- (1.8,0);
            \draw (1.3,0.7) -- (1.8,0.7);
            \draw (1.3,2.1) -- (1.8,2.1);
            % Draw the second batch of E tensors:
            \ETensor{2.8,0}{$\mathbb{E}$}{purple}{0.5}{0.7}
            \ETensor{2.8,1.4}{$\mathbb{E}$}{purple}{0.5}{0.7}
            % Draw second batch of swaps:
            \draw (3.8,0) -- (4.3,1.4) -- (5.2,1.4);
            \draw (3.8,1.4) -- (4.3,0) -- (5.2,0);
            \draw (3.8,0.7) -- (5.2,0.7);
            \draw (3.8,2.1) -- (5.2,2.1);
            % Now draw parentheses:
            \draw [decorate, decoration = {calligraphic brace,mirror}] (-1,2.3) --  (-1,-0.2);
            \draw [decorate, decoration = {calligraphic brace}] (4.4,2.3) --  (4.4,-0.2);
            \draw (4.8,2.4) node {\scriptsize $n$};
        \end{tikzpicture}
    \end{array} ]
     \nonumber \\
    &\xrightarrow{\scriptsize k\to\infty}
    \left(
    \begin{array}{c}
		\begin{tikzpicture}[scale=.45, baseline={([yshift=0ex]current bounding box.center)}, thick]
            % The E tensors and their unions on the left:
            \ETensor{0,0}{$\mathbb{E}$}{purple}{0.5}{0.7}
            \ETensor{0,1.4}{$\mathbb{E}$}{purple}{0.5}{0.7}
            \draw (-1,0.7) -- (-1,1.4);
            \draw (-1,0) -- (-1.3,0) -- (-1.3,2.1) -- (-1,2.1);
            % Draw the swaps:
            \draw (1,0) -- (1.5,1.4);
            \draw (1,1.4) -- (1.5,0);
            \draw (1,0.7) -- (1.5,0.7);
            \draw (1,2.1) -- (1.5,2.1);
            % Draw the right fixed points:
            \RightVector{2.5,0}{$\Lambda$}{yellow}{0.5}{0.7}
            \RightVector{2.5,1.4}{$\Lambda$}{yellow}{0.5}{0.7}
        \end{tikzpicture}
    \end{array}
    \right) ^n =: \eta^n . \label{eq:def_A(D)}
\end{align} 
where $\eta$ is a positive constant ($0 < \eta \leq 1$) that only depends on the properties of tensor $A$. By combining this with Eq. (\ref{eq:purity1}), we obtain
\begin{align} 
    \Tr[(\rho_S)^2] &\xrightarrow{k\to\infty} \eta^n \geq \frac{1}{D} - 4\varepsilon \label{eq:purity2}, 
\end{align}
where the lower bound is a strictly positive quantity if $\varepsilon < \frac{1}{4D}$. This leads to a contradiction if $\eta < 1$, so it must necessarily hold that $\eta = 1$. Note that $\eta$ is the purity of a one-body density matrix $\rho_1$, i.e. $\eta = \Tr[(\rho_1)^2]$ for
\begin{align*}
        \rho_1 := 
        \begin{tikzpicture}[scale=.45, baseline={([yshift=-1ex]current bounding box.center)}, thick]
            \MPSTensor{0,1.75}{$A$}{purple}
            \MPSTensordown{0,-0.25}{$A^*$}{purple}
            \simplematrixvertical{1.5,0.75}{$\Lambda$}{yellow}
            \draw (0.9,1.75) -- (1.5,1.75);
            \draw (0.9,-0.25) -- (1.5,-0.25);
            \draw (-1,-0.25) -- (-1,1.75);
        \end{tikzpicture}
        = \begin{tikzpicture}[scale=.45, baseline={([yshift=-1ex]current bounding box.center)}, thick]
            \MPSTensor{0,1.5}{$C$}{purple}
            \MPSTensordown{0,0}{$C^\dagger$}{purple}
            \draw (-1,0) -- (-1,1.5);
            \draw (1,0) -- (1,1.5);
        \end{tikzpicture}
        \ ,
        \text{ with }
        \begin{tikzpicture}[scale=.45, baseline={([yshift=-1ex]current bounding box.center)}, thick]
            \MPSTensor{0,0}{$C$}{purple}
        \end{tikzpicture}
        := \begin{tikzpicture}[scale=.45, baseline={([yshift=-1ex]current bounding box.center)}, thick]
            \MPSTensor{0,0}{$A$}{purple}
            \simplematrix{1.5,0}{$\Lambda^{\frac{1}{2}}$}{yellow}
        \end{tikzpicture}.
    \end{align*}
This can be checked as follows,
\begin{equation*}
    \Tr[(\rho_1)^2] = 
    \begin{tikzpicture}[scale=.45, baseline={([yshift=-1ex]current bounding box.center)}, thick]
        \begin{scope}[shift={(0,0)}]
            \MPSTensor{0,1.75}{$A$}{purple}
            \MPSTensordown{0,-0.25}{$A^*$}{purple}
            \simplematrixvertical{1.5,0.75}{$\Lambda$}{yellow}
            \draw (0.9,1.75) -- (1.5,1.75);
            \draw (0.9,-0.25) -- (1.5,-0.25);
            \draw (-1,-0.25) -- (-1,1.75);
        \end{scope}
        \begin{scope}[shift={(0,3.5)}]
            \MPSTensor{0,1.75}{$A$}{purple}
            \MPSTensordown{0,-0.25}{$A^*$}{purple}
            \simplematrixvertical{1.5,0.75}{$\Lambda$}{yellow}
            \draw (0.9,1.75) -- (1.5,1.75);
            \draw (0.9,-0.25) -- (1.5,-0.25);
            \draw (-1,-0.25) -- (-1,1.75);
        \end{scope}
        % La linea de la traza:
        \draw (0,6.25) -- (-1.5,6.25) -- (-1.5,-1.25) -- (0,-1.25);
    \end{tikzpicture}
    = \
    \begin{tikzpicture}[scale=.45, baseline={([yshift=-1ex]current bounding box.center)}, thick]
        \begin{scope}[shift={(0,0)}]
            \MPSTensordown{0,1.5}{$A$}{purple}
            \MPSTensor{0,0}{$A^*$}{purple}
        \end{scope}
        \begin{scope}[shift={(0,3)}]
            \MPSTensordown{0,1.5}{$A$}{purple}
            \MPSTensor{0,0}{$A^*$}{purple}
        \end{scope}
        % Las lineas de la izquierda:
        \draw (-1,1.5) -- (-1,3);
        \draw (-1,0) -- (-1.5,0) -- (-1.5,4.5) -- (-1,4.5);
        % Ahora el Lambda de arriba:
        \begin{scope}[shift={(2.5,0.75)}]
    		\filldraw[fill=yellow] (-1/2,-1) -- (-1/2,1) -- (1/2,1) -- (1/2,-1) -- (-1/2,-1);
    		\draw (0,0) node {\scriptsize $\Lambda$};
    	\end{scope}
        % Ahora el Lambda de abajo:
        \begin{scope}[shift={(2.5,3.75)}]
    		\filldraw[fill=yellow] (-1/2,-1) -- (-1/2,1) -- (1/2,1) -- (1/2,-1) -- (-1/2,-1);
    		\draw (0,0) node {\scriptsize $\Lambda$};
    	\end{scope}
        % Ahora las lineas cruzadas:
        \draw (1,0) -- (1.5,3) -- (2,3);
        \draw (1,3) -- (1.5,0) -- (2,0);
        % Ahora las lineas rectas:
        \draw (1,1.5) -- (2,1.5);
        \draw (1,4.5) -- (2,4.5);
    \end{tikzpicture}
    \ = 
    \eta.
\end{equation*}

Therefore, we have that $\eta = 1$ implies that $\rho_1$ is a pure state and $\text{rank}(CC^\dagger) = 1$, meaning that $C = \dyad{u}{v}$ for some $\ket{u} \in \mathbb{C}^{D^2}$, $\ket{v} \in \mathbb{C}^{d}$,
    \begin{equation*}
        \begin{tikzpicture}[scale=.45, baseline={([yshift=-1ex]current bounding box.center)}, thick]
            \MPSTensor{0,0}{$C$}{purple}
        \end{tikzpicture}
        := \begin{tikzpicture}[scale=.45, baseline={([yshift=-1ex]current bounding box.center)}, thick]
            \draw (-1,0) -- (1,0);
            \draw (0,0.65) -- (0,1.5);
            \draw[thick, fill=yellow] (0,0) circle (0.25);
            \draw[thick, fill=amaranth] (0,0.65) circle (0.25);
        \end{tikzpicture}.
    \end{equation*}
    Then, using the fact that $\Lambda$ is invertible, we have
    \begin{align*}
        \ket{\psi_N(A)} &= 
        \begin{tikzpicture}[scale=.45, baseline={([yshift=-1ex]current bounding box.center)}, thick]
            \FullMPS{(0,0)}{$A$}{purple}{yellow}
        \end{tikzpicture}
        \\
        &= 
        \begin{tikzpicture}[scale=.45, baseline={([yshift=-1ex]current bounding box.center)}, thick]
            \MPSTensor{0,0}{$C$}{purple}
            \rectangularsimplematrix{1.65,0}{ $\Lambda^{-\frac{1}{2}}$}{yellow}{0.15}
            \MPSTensor{3.3,0}{$C$}{purple}
            \rectangularsimplematrix{4.95,0}{ $\Lambda^{-\frac{1}{2}}$}{yellow}{0.15}
            \draw[shift={(0,0)},dotted] (5.95,0) -- (7.2,0);
            \draw (7.2,0) -- (7.7,0) -- (7.7,-0.8) -- (-1,-0.8) -- (-1,0);
        \end{tikzpicture}
        \\
        &= K \cdot 
        \begin{tikzpicture}[scale=.45, baseline={([yshift=-1ex]current bounding box.center)}, thick]
            \draw[thick, fill=amaranth] (0,0) circle (0.25);
            \draw (0,0.25) -- (0,0.75);
            \draw[thick, fill=amaranth] (0.75,0) circle (0.25);
            \draw (0.75,0.25) -- (0.75,0.75);
            \draw[thick, fill=amaranth] (1.5,0) circle (0.25);
            \draw (1.5,0.25) -- (1.5,0.75);
            \draw[shift={(0,0)},dotted] (2,0) -- (3,0);
            \draw[thick, fill=amaranth] (3.5,0) circle (0.25);
            \draw (3.5,0.25) -- (3.5,0.75);
            \draw[thick, fill=amaranth] (4.25,0) circle (0.25);
            \draw (4.25,0.25) -- (4.25,0.75);
        \end{tikzpicture}
    \end{align*}
    for some $K \in \mathbb{C}$. Therefore, $D_\mathds{1} = 1$ and $\ket{\psi_N(A)}$ is a product state for all $N$. 
\end{proof}

Note that, if $\varepsilon$ and $D$ are allowed to depend on $N$, with $D_N = O(\text{poly}(N))$, so that $\ket{\psi_N(A)}$ has the MPS-up$_{\varepsilon_N, D_N}$ property, then $\varepsilon_N < \frac{1}{4D_N}$ is not sufficient, since we could have for example $\varepsilon_N = \frac{1}{4D_N} - \frac{1}{8} \eta^N$. Nevertheless, if we additionally require the existence of a positive sequence $(g_N)$ with $g_N = \Omega(1/\text{poly}(N))$, such that $\varepsilon_N < \frac{1}{4 D_N} - g_N$, then Eq. (\ref{eq:purity2}) still leads to a contradiction and the claim holds.

Finally, we can effectively remove the normality assumption, using Prop. \ref{prop:purity_normal} together with the block-injectivity property. We obtain that any approximate TI MPS-up has a GHZ-type of entanglement, given that $\varepsilon$ is small enough. This is summarized in the following theorem, whose proof is provided in Appendix \ref{app:proof_puritynonnormal}.

\begin{mythm}{2}[Approximate MPS-up] \label{thm:approx_MPS-up_thm}
    Let $\{\ket{\psi_N(A)}\}_N$ be a family of TI MPS with the approximate MPS-up$_{\varepsilon_N, D_N}$ property for all $N$ larger than some $N_0$, where $D_N = O(\text{poly}(N))$. 
    Let $b$ be the number of elements in the BNT of tensor $A$.
    Then, if there exists a positive sequence $(g_N)$ with $g_N = \Omega(1/\text{poly}(N))$, such that either
    \begin{enumerate}[(a)]
        \item $b = 1$ and $0 \leq \varepsilon_N < \frac{1}{4D_N}-g_N$, or
        \item $b > 1$ and $0 \leq \varepsilon_N <  \left( \frac{1}{4D_N} - g_N \right) \frac{\min_i |\alpha_{i}|}{2(\sum_{i} |\alpha_{i}|^2)^{\frac{1}{2}}}$, 
    \end{enumerate}
    where $\alpha_i$ denote the coefficients weighting each element of the BNT of $A$ according to Eq. \eqref{eq:nonnormal_MPS_intro}, then $\ket{\psi_N(A)} = \sum_{i=1}^b \beta_i \ket{\phi_i}^{\otimes N}$.
\end{mythm}
Note that, for the normal case ($b = 1$) with $D_N = O(1)$, the bound becomes $\varepsilon < \frac{1}{8D}$, which is tighter than the bound $\varepsilon < \frac{1}{4D}$ found in Prop. \ref{prop:purity_normal}. Furthermore, the result carries over to the fixed-$N$ setting. In particular, by choosing a single permutation $\pi$ (depicted in Eq. \eqref{eq:permutation_approx_case}) with $k = \lfloor \sqrt{N} \rfloor$, one still reaches the desired contradiction in Eq. \eqref{eq:purity2} unless the state under consideration is a product state, provided that $N$ is sufficiently large.

\section{Non-translationally invariant MPS-up}
\label{sec:non-TI}

We consider now the more general setting of non-TI MPS with the MPS-up$_{\varepsilon, D}$ property. We will use $L$ to denote $L := \min \{ l \in \mathbb{N} \mid D^2 \leq d^l \}$. Note that, given a generic non-TI MPS, blocking every $L$ consecutive sites yields an injective tensor almost surely. The reason for this is that, among all linear maps $f : \mathbb{C}^{D^2} \to \mathbb{C}^{d^L}$ with $D^2 \leq d^L$, those that fail to be injective correspond to rank-deficient matrices, which form a measure-zero set of the full matrix space.

The following proposition shows that generic non-TI MPS with the exact MPS-up property are products of almost all injective blocks of $L$ sites.

\begin{proposition}[Exact non-TI MPS-up]
    \label{prop:rank_counting_non-uniform_injective} 
    An injective non-TI MPS with the exact MPS-up$_{D}$ property on a sufficiently large number $N$ of particles is a product of almost all sites, meaning that it has to satisfy $\prod_{i = i_0}^{N} D_{[i]} \leq D$, where $i_0 = 1$ for even $N$ and $i_0 = 2$ for odd $N$. 
\end{proposition}
\begin{proof}
    We can follow the same steps as in Lemma \ref{lemma:rank_counting_injective}, changing the definition of $\mathcal{A}^{-1}$ to $\mathcal{A}^{-1} := (A^{[1]})^{-1} \otimes \dots \otimes (A^{[N]})^{-1}$. The Schmidt rank across the half-chain bipartition of $U_\pi \mathcal{A}^{-1} \ket{\psi}$, where $\pi$ is the permutation in Eq. (\ref{eq:def_permutation_rank_counting}), is equal to $\prod_{i=i_0}^{N} D_{[i]}$ where $i_0 = 1$ for even $N$, or $i_0 = 2$ for odd $N$. Then, the inequality $\prod_{i = i_0}^{N} D_{[i]} \leq D$ must hold. Informally speaking, this implies that $D_{[i]} = 1$ for many values of $i$ when $D$ is small with respect to $N$, or in other words, $\ket{\psi}$ should be a product of almost all the particles.
\end{proof}

The same conclusion holds for non-TI MPS when the MPS-up property holds only in an approximate sense and exponential decay of correlations is expected. To exclude pathological ``near-singular'' MPS, where the slightest perturbation could destroy injectivity and induce long-range correlations, we make the technical assumption that each injective block of tensors is \textit{well-conditioned}, i.e. $\| A_{[i]}^{-1} \|_{\text{op}} = O(1)$. We provide the proof in Appendix \ref{app:non-TI-MPSup}.
\begin{myprop}{4}[Approximate non-TI MPS-up] \label{prop:purity_non-TI_injective}
    A non-TI MPS-up$_{\varepsilon, D}$ with exponentially decaying correlations, well-conditioned tensors and $\varepsilon < \frac{1}{4D}$, on a sufficiently large number of particles, can be written as a product of almost all sites or blocks of a few sites.
\end{myprop}

Our results can also be extended to non-generic block-injective non-TI MPS, under a suitably strong set of assumptions. For instance, for the exact MPS-up property, it is sufficient to assume that the blocks of the tensors at all sites live in common subspaces, both at the virtual level and at the physical level. For the approximate MPS-up property, additionally requiring ergodicity for each separate sequence of blocks would be enough. This way, one can conclude that the state should be of the form $\ket{\psi} = \sum_{j=1}^b \gamma_j \ket{\psi_j}$, where each $\ket{\psi_j}$ is a product of almost all sites.

\section{From MPS to product-state ansätze}
\label{sec:implications_mps_approx_algorithms}

In this section, we discuss the algorithmic implications of our results in settings without a 1D geometry. In such “geometry-agnostic” problems, we argue in favor of starting directly with a variational ansatz consisting of a product state or a small superposition thereof, rather than an MPS, as it is simpler and more efficient.

% As a concrete example where the MPS ansatz can be replaced by a superposition of a few product states, consider a permutationally invariant Hamiltonian on $N$ qubits whose ground state $\ket{\psi_0}$ is an MPS of low bond dimension. Since it lies in the symmetric subspace, we can write $\ket{\psi_0} = \sum_{p=0}^k d_p \ket{D_N^p}$ for some $d_p \in \mathbb{C}$, where $k \leq N$ and $\ket{D_N^p}$ denotes the $N$-qubit Dicke state with $p$ excitations. Each $\ket{D_N^p}$ admits a non-TI MPS representation with bond dimension $p+1$ \cite{raveh_2024_dicke-as-MPS}, and can also be approximated arbitrarily well by a sum of $p+1$ TI product states of the form $\ket{\phi}^{\otimes N}$ \cite{Aloy_2021_symmetric-states}. As a result, $\ket{\psi_0}$ has an MPS representation of bond dimension at most $(k+1)(k+2)/2$, and can likewise be approximated by that many TI product states. Since $\ket{\psi_0}$ is an MPS, we are in a low-entanglement regime and thus $k \ll N$. This means that instead of working with a full MPS, it is enough to use only $(k+1)(k+2)/2$ product states in order to approximate $\ket{\psi}$ to any desired accuracy. \alv{here we can talk about the Lipkin model \cite{Latorre_2005}}

%We now move beyond the fully symmetric setting to consider more general cases without permutational invariance. In what follows, we show that our results rigorously justify replacing the MPS ansatz with a simpler product-type one when certain assumptions and the MPS-up property hold. 

To demonstrate this idea, let us first define a weaker notion of MPS-up that is both sufficient and efficiently verifiable. Given a certain permutation $\pi \in \mathcal{S}_N$, we say that $\ket{\psi_1}$ is an \textit{MPS-up$_{\pi, \varepsilon, D}$} if there is an MPS $\ket{\psi_2}$ of bond dimension $\leq D$ such that
\begin{equation} \label{eq:inequality_implications}
    \| U_\pi \ket{\psi_1} - \ket{\psi_2} \| \leq \varepsilon .
\end{equation}
If $\pi$ is carefully chosen, then the MPS-up$_{\pi, \varepsilon, D}$ property is sufficient to establish every entry in Table \ref{table:results}. Indeed, in each of our proofs it is enough to invoke the inequality in Eq. \eqref{eq:inequality_implications} just for one of the following permutations:
\begin{itemize}
    \item Exact MPS-up: a permutation $\pi$ that moves all odd-indexed particles to the left and even-indexed ones to the right (see Eq. \eqref{eq:def_permutation_rank_counting}).
    \item Approximate MPS-up: a permutation $\pi$ that brings to the front those sites whose positions are multiples of $\lfloor\sqrt{N}\rfloor$, given sufficiently large $N$ (depicted in Eq. \eqref{eq:permutation_approx_case} with $k = \lfloor\sqrt{N}\rfloor$), since in this case Eq. \eqref{eq:purity2} would still lead to the desired contradiction unless $\ket{\psi_1}$ is a product state.
\end{itemize}
Hence, there is no need to verify that the state under study is an MPS for all $N!$ permutations, as required by the full MPS-up$_{\varepsilon, D}$ property: checking it for a well-chosen $\pi$ already guarantees our conclusions.

Let $\mathcal{H}$ be a Hamiltonian with a gap $\Delta > 0$ on $N$ particles, with unique ground state $\ket{\psi_0}$ such that $\mathcal{H} \ket{\psi_0} = 0$. Pick a special permutation $\pi$ of the form described above, and let $\ket{\psi_1}$ and $\ket{\psi_2}$ be two bond dimension-$D$ MPS obtained when trying to approximate the ground state with two different orderings of the particles, $\{1, \dots, N\}$ and $\{\pi(1), \dots, \pi(N)\}$, respectively, each with variational energy $\leq \varepsilon$. Then, we have that  
\begin{equation} \label{eq:cor_implications_aux1}
   \| U_\pi \ket{\psi_1} - \ket{\psi_2} \|
   \leq 2 \sqrt{\varepsilon/\Delta},
\end{equation}
by Lemma \ref{lemma:technical_algorithms} of Appendix \ref{sec:app_algorithms}, so $\ket{\psi_1}$ is an MPS-up$_{\pi, \tilde{\varepsilon}, D}$ with $\tilde{\varepsilon} :=2\sqrt{\varepsilon/\Delta}$, and the following corollary about $\ket{\psi_0}$ can be shown, whose proof we include in Appendix \ref{sec:app_algorithms}. 
\begin{mycly}{5} \label{cor:mpsalgorithms}
    Under the above assumptions, the true ground state $\ket{\psi_0}$ of $\mathcal{H}$ is at most $2\sqrt{\varepsilon/\Delta}$ away from one of the simple ansätze listed in the box of Table \ref{table:results} that corresponds to the structure assumed on $\ket{\psi_1}$, provided $\varepsilon/\Delta$ is small enough (i.e. $\ket{\psi_0}$ is a product state, a small superposition thereof, or a product of short-range entangled blocks).
\end{mycly}
Thus, if two distinct MPS orderings as described both achieve good energies, the only way this can happen is if the ground state is already close to a simple product-based structure, so that MPS are effectively unnecessary. Consequently, when no intrinsic one-dimensional geometry is present and only limited entanglement is expected, using a product-type variational ansatz instead of MPS can yield substantial savings in both runtime and memory without sacrificing accuracy.

\section{Almost-product MPS-under-permutations} \label{sec:almost-product}

For completeness, we explore some examples of MPS-under-permutations that lie beyond the scope of our assumptions. These examples include the paradigmatic W state $\ket{W_N}$ \cite{dur2000three} and Dicke states $\ket{D_{n,N}}$ \cite{dicke1954coherence}, along with a subclass $\ket{\chi_{a,N}}$ of the so-called weight states \cite{christandl_optimization_2021}. Given
\begin{align*}
    \ket{\chi_{a,\delta,N}} := \sum_{\substack{i_1+\dots+i_N=a\\ i_1, \dots, i_N \in \{0,1,\dots, \delta\}}} \ket{i_1 i_2 \dots i_N}, 
\end{align*}
they can be written as $\ket{W_N} := \ket{\chi_{1,1,N}}$, $\ket{D_{n,N}} := \ket{\chi_{n,1,N}}$ and $\ket{\chi_{a,N}} := \ket{\chi_{a,a,N}}$. 
%\newpage

They are permutationally invariant and admit minimal MPS representations of the form
\begin{equation} \label{eq:MPS-Xform}
    \ket{\psi_N(X, A)} =
    \begin{array}{c}
		\begin{tikzpicture}[scale=.4, baseline={([yshift=-3ex]current bounding box.center)}, thick]
            \FullMPSX{(0,0)}{$A$}{$X$}{purple}{yellow}
        \end{tikzpicture}
    \end{array},
\end{equation}
whose bond dimensions are specified in the first column of Table \ref{table:non-TI}. However, these representations are non-TI and non-injective, and any attempt to express them with a TI MPS would necessarily result in a bond dimension scaling with the system size, so our previous results cannot be applied in these cases (see 2nd column in Table \ref{table:non-TI}). 

In the corollary below, we prove that the lower bound of $\Omega(N^{1/(3+\delta)})$ for the bond dimension of TI MPS representations that was previously shown for $\ket{W_N}$ and $\ket{\chi_{a,N}}$ in the literature, also holds for a broader range of states, including $\ket{D_{n,N}}$. The proof is included in Appendix \ref{sec:app_DTI_bound}. 
\newpage 
\begin{restatable}{corollary}{lowerboundDTI} \label{cor:lowerbound_DTI}
    Given any family $\{\ket{\psi_N(A^{(N)})}\}_N$ of exact TI MPS-under-permutations whose tensor rank increases with $N$, the bond dimension of any TI MPS representation has to satisfy $D_\text{TI} = \Omega(N^{1/(3+\delta)})$ for each $\delta > 0$.
\end{restatable}

In fact, none of these states can be expressed as a product state or a superposition of a few of them, since their tensor rank scales linearly with $N$ (see 3rd column in Table \ref{table:non-TI}). Thus, even though they have the exact MPS-up$_{D}$ property, the implications of our theorems are violated, as they cannot be written exactly as a superposition of a constant number of product states for all $N$.

Despite this, they can be effectively approximated by a sum of a constant number of product states up to arbitrary accuracy (see 4th column in Table \ref{table:non-TI}), so their entanglement structure still has an almost-product nature, similar to the cases studied in previous sections. Whether this conclusion holds for all MPS-under-permutations beyond these examples remains an open question. 

\begin{table}[t]
    \begin{minipage}{\textwidth}
    \centering
    %\centering
    \begin{tabularx}{0.73\textwidth}{c||c|c|c|c} 
         MPS-up $\ket{\psi}$ & $D$ & $D_\text{TI}$ & rk($\ket{\psi}$) & \underline{rk}($\ket{\psi}$)  \\ \hline\hline
         $\ket{W_N}$ & 2 & $\Omega(N^{\frac{1}{3+\delta}})$ \cite{Perez-Garcia2007, Michalek2018} & $N$ & 2  \\ \hline
         $\ket{D_{n,N}}$ & $\min\{n,N-n\} + 1$ & \makecell{$\Omega(N^{\frac{1}{3+\delta}})$ \\ (Appendix \ref{sec:app_DTI_bound})} & \makecell{$\max\{n,N-n\} + 1$ \\ (Thm. 3 in \cite{chen_tensor_2010})} & \makecell{$\min\{n,N-n\} + 1$ \\ (section IV \cite{vrana_asymptotic_2015})} \\ \hline
         $\ket{\chi_{a,N}}$ & \makecell{$a+1$ \\ (Lemma 3 \cite{christandl_optimization_2021})} & \makecell{$\Omega(N^{\frac{1}{3+\delta}})$ \\ (Lemma 5 \cite{christandl_optimization_2021})} & \makecell{$\geq N+1$ \\ \cite{christandl_optimization_2021}} & \makecell{$a+1$ \\ (Prop. 2 \cite{christandl_optimization_2021})} \\ \hline  
    \end{tabularx}
    \captionsetup{justification=raggedright}
    \caption[foo bar]{ 
    Quantities of interest of the MPS-under-permutations studied in section \ref{sec:almost-product}. Note that $\ket{W_N} = \ket{D_{1,N}} = \ket{\chi_{1,N}}$.
    \linebreak
    $\bullet$ \textbf{1st column:} bond dimension of a non-TI MPS representation of the form in Eq. (\ref{eq:MPS-Xform}). $\ket{W_N}$ admits the MPS $A^0 = \mathds{1}_{2}, A^1 = \dyad{1}{2}, X = \dyad{2}{1}$; given $n \leq \frac{N}{2}$, $\ket{D_{n,N}}$ admits the MPS $A^0 = \mathds{1}_{n+1}$, $A^1 = \sum_{j=1}^{n} \dyad{j}{j+1}$, $X = \sum_{j=1}^{n} \dyad{j+1}{j}$.
    \linebreak
    $\bullet$ \textbf{2nd column:} lower bound on the minimal bond dimension if translational invariance is enforced. Previously known bounds were $\Omega(N^{1/(3+\delta)})$ for $\ket{W_N}$ \cite{Perez-Garcia2007, Michalek2018} and for $\ket{\chi_{a,N}}$ (Lemma 5 \cite{christandl_optimization_2021}). We show the improved lower bound for $\ket{D_{n,N}}$ in Corollary \ref{cor:lowerbound_DTI}.
    \linebreak
    $\bullet$ \textbf{3rd column:} rk($\ket{\psi}$) $:= \min\{r : \ket{\psi} = \sum_{i=1}^r \ket{\phi_i}, \ \ket{\phi_i}$ product state$\}$ is known as the \textit{tensor rank}.
    \linebreak
    $\bullet$ \textbf{4th column:} \underline{rk}($\ket{\psi}$) $:= \min\{r : \ket{\psi} = \lim_{\varepsilon \to 0} \frac{1}{\varepsilon^e} \left( \sum_{i=1}^r \ket{\phi_i(\varepsilon)} \right), \ \ket{\phi_i(\varepsilon)}$ product state$\}$ is known as the \textit{border rank}, and it is lower bounded by the maximal Schmidt rank across any bipartition. The W state has Schmidt rank 2 and can be written $\ket{W_N} = \lim_{\varepsilon \to 0} \frac{1}{2\varepsilon} ([\ket{0} + \varepsilon \ket{1}]^{\otimes N} - [\ket{0} - \varepsilon \ket{1}]^{\otimes N})$, so \underline{rk}$(\ket{W_N}) = 2$. See \cite{landsberg2011tensors} for more details on rk($\ket{\psi}$) and \underline{rk}($\ket{\psi}$).}
    \label{table:non-TI}
    \end{minipage}
\end{table}

%\balancecolsandclearpage

\section{Outlook}

We have shown how MPS that are stable under permutations, meaning that any re-arrangement of their particles admits an efficient MPS representation, are at best trivially so, since this is only possible if they are product states or simple combinations thereof depending on the properties of the tensors.

In the TI setting, we prove that states with the MPS-under-permutation property, both in an exact and an approximate sense, are superpositions of as many product states as the number of elements in their basis of normal tensors. In the non-TI setting, we show that they are products of most of their sites, under the additional assumptions of injectivity, as well as ergodicity for the approximate setting. With this, we further highlight how the underlying geometric structure of a problem is the key factor in establishing non-trivial MPS as the preferred ansatz over simpler and less entangled alternatives.

Our results exhaustively characterize the product nature of MPS-under-permutations admitting efficient TI MPS representations, as well as generic non-TI MPS-under-permutations. Moreover, relevant states like the W and Dicke states, which do not fit into these assumptions, can be accurately approximated by such product ansätze. It remains an open question whether this conclusion holds for all MPS-under-permutations. Using the concepts of tensor and border rank, which come from algebraic complexity theory, and have been widely used for the study of multipartite entanglement \cite{chitambar_tripartite_2008, chen_tensor_2010,vrana_asymptotic_2015, christandl2023resource}, this problem can be rephrased as finding whether there exists a family of MPS-under-permutations whose border rank increases with the system size.

\section*{Acknowledgements}

MFL acknowledges support from the International Max Planck Research School for Quantum Science and Technology (IMPRS-QST). This research is part of the Munich Quantum Valley (MQV), which is supported by the Bavarian State Government with funds from the High-tech Agenda Bayern Plus. AMA acknowledges support from the Spanish Agencia Estatal de Investigacion through the grants ``“IFT Centro de Excelencia Severo Ochoa CEX2020-001007-S" and ``Ram\'on y Cajal RyC2021-031610-I'', financed by MCIN/AEI/10.13039/501100011033 and the European Union NextGenerationEU/PRTR. DPG acknowledges support from the Spanish Ministry of Science and Innovation MCIN/AEI/10.13039/501100011033 (grants CEX2023-001347-S and PID2020-113523GB-I00). This work has been financially supported by the Ministry for Digital Transformation and of Civil Service of the Spanish Government through the QUANTUM ENIA project call - Quantum Spain project, and by the European Union through the Recovery, Transformation and Resilience Plan – NextGenerationEU within the framework of the Digital Spain 2026 Agenda. This research was supported in part by Perimeter Institute for Theoretical Physics. Research at Perimeter Institute is supported by the Government of Canada through the Department of Innovation, Science and Economic Development and by the Province of Ontario through the Ministry of Research, Innovation and Science.
NS acknowledges support from the European Union’s Horizon 2020 research and innovation programme through Grant No.\ 863476 (ERC-CoG SEQUAM),  the Austrian Science Fund FWF (Grant DOIs \href{https://doi.org/10.55776/COE1}{10.55776/COE1}, \href{https://doi.org/10.55776/P36305}{10.55776/P36305}, and \href{https://doi.org/10.55776/F71}{10.55776/F71}),  and the European Union -- NextGenerationEU.

\bibliography{apssamp.bib}

\appendix

\section{Proofs for exact MPS-under-permutations}

\subsection{Exact MPS-up with normal tensor} \label{app:proof_rankcountingnormal}

\rankcountingnormal*
\begin{proof}
    The proof is based on a rank counting argument after a certain permutation is performed. Let $\tilde{N} := \lfloor N/L_I \rfloor$, where $L_I$ is the injectivity length defined in section \ref{sec:TNs-background}, and let $\pi \in \mathcal{S}_{\tilde{N}}$ be defined as
    {\small \begin{equation} \label{eq:def_permutation_rank_counting}
        \begin{cases}
            \pi(2k-1) = k, &\text{for } k = 1, \dots, \lceil \frac{\tilde{N}}{2} \rceil \\
            \pi(2k) = \lceil \frac{\tilde{N}}{2} \rceil + k , &\text{for } k = 1, \dots, \lfloor \frac{\tilde{N}}{2} \rfloor
        \end{cases}.
    \end{equation}}
    An example for $\tilde{N} = 7$ can be depicted as 
    \begin{equation*} \begin{tikzpicture}[scale=.45, baseline={([yshift=-3ex]current bounding box.center)}, thick]
        \foreach \x in {1,2,3,4,5,6,7}{
            \node[below] at (\x-1,-0.1) {\scriptsize $\x$};
		}
        \draw[thick, fill=purple] (0,0) circle (0.25);
        \draw[thick, fill=darkpurple] (1,0) circle (0.25);
        \draw[thick, fill=purple] (2,0) circle (0.25);
        \draw[thick, fill=darkpurple] (3,0) circle (0.25);
        \draw[thick, fill=purple] (4,0) circle (0.25);
        \draw[thick, fill=darkpurple] (5,0) circle (0.25);
        \draw[thick, fill=purple] (6,0) circle (0.25);
        \node at (7,0) {$\xrightarrow{\pi}$};
        \draw[thick, fill=purple] (8+0,0) circle (0.25);
        \node[below] at (8,-0.1) {\scriptsize 1};
        \draw[thick, fill=purple] (8+1,0) circle (0.25);
        \node[below] at (9,-0.1) {\scriptsize 3};
        \draw[thick, fill=purple] (8+2,0) circle (0.25);
        \node[below] at (10,-0.1) {\scriptsize 5};
        \draw[thick, fill=purple] (8+3,0) circle (0.25);
        \node[below] at (11,-0.1) {\scriptsize 7};
        \draw[thick, fill=darkpurple] (8+4,0) circle (0.25);
        \node[below] at (12,-0.1) {\scriptsize 2};
        \draw[thick, fill=darkpurple] (8+5,0) circle (0.25);
        \node[below] at (13,-0.1) {\scriptsize 4};
        \draw[thick, fill=darkpurple] (8+6,0) circle (0.25);
        \node[below] at (14,-0.1) {\scriptsize 6};
    \end{tikzpicture} . \end{equation*}

    Let $0 \leq r < L_I$ be the unique integer such that $N = L_I \tilde{N} + r$. Then, we block the tensors of the MPS such that we end up with $r$ blocks of $L_I + 1$ sites, and $\tilde{N}-r$ blocks of $L_I$ sites. For instance, for $N = 14$ and $L_I = 3$, we would block the MPS in the following way,
    \begin{align*}
        &\begin{array}{c} \begin{tikzpicture}[scale=.45,thick,baseline={([yshift=0ex]current bounding box.center)}] 
            \draw (0.25,0) -- (14.75,0) -- (14.75,-0.75) -- (0.25,-0.75) -- (0.25,0);
            \foreach \x in {1,2,3,4,5,6,7,8,9,10,11,12,13,14}{
                \filldraw[fill=purple] (\x-0.25,-0.25) -- (\x-0.25,0.25) -- (\x+0.25,0.25) -- (\x+0.25,-0.25) -- (\x-0.25,-0.25);
                \draw (\x, 0.25) -- (\x, 0.75);
    		}
            \draw[densely dashed,draw=red] (0.6,0.5) -- (4.4,0.5) -- (4.4,-0.5) -- (0.6,-0.5) -- (0.6,0.5);
            \draw[densely dashed,draw=red] (4.6,0.5) -- (8.4,0.5) -- (8.4,-0.5) -- (4.6,-0.5) -- (4.6,0.5);
            \draw[densely dashed,draw=red] (8.6,0.5) -- (11.4,0.5) -- (11.4,-0.5) -- (8.6,-0.5) -- (8.6,0.5);
            \draw[densely dashed,draw=red] (11.6,0.5) -- (14.4,0.5) -- (14.4,-0.5) -- (11.6,-0.5) -- (11.6,0.5);
        \end{tikzpicture} \end{array} \\
        &\to \begin{array}{c} \begin{tikzpicture}[scale=.45,thick,baseline={([yshift=0ex]current bounding box.center)}] 
            \draw (0.25,0) -- (5.1,0) -- (5.1,-0.6) -- (0.25,-0.6) -- (0.25,0);
            \foreach \x in {1,2.25}{
                \filldraw[fill=amaranth] (\x-0.45,-0.25) -- (\x-0.45,0.25) -- (\x+0.45,0.25) -- (\x+0.45,-0.25) -- (\x-0.45,-0.25);
    		}
            \foreach \x in {3.4,4.45}{
                \filldraw[fill=yellow] (\x-0.35,-0.25) -- (\x-0.35,0.25) -- (\x+0.35,0.25) -- (\x+0.35,-0.25) -- (\x-0.35,-0.25);
    		}
            \foreach \center in {1,2.25}{
                \draw (\center-0.15, 0.25) -- (\center-0.15, 0.65);
                \draw (\center-0.05, 0.25) -- (\center-0.05, 0.65);
                \draw (\center+0.05, 0.25) -- (\center+0.05, 0.65);
                \draw (\center+0.15, 0.25) -- (\center+0.15, 0.65);
            }
            \foreach \center in {3.4,4.45}{
                \draw (\center, 0.25) -- (\center, 0.65);
                \draw (\center+0.1, 0.25) -- (\center+0.1, 0.65);
                \draw (\center-0.1, 0.25) -- (\center-0.1, 0.65);
            }
        \end{tikzpicture} \end{array}.
    \end{align*}
    Note that the resulting MPS consists of two different injective tensors, $A : \mathcal{M}_D(\mathbb{C}) \to \mathbb{C}^{d^{L_I+1}}$ and $\tilde{A} : \mathcal{M}_D(\mathbb{C}) \to \mathbb{C}^{d^{L_I}}$ with the same bond dimension $D_\mathds{1}$, whose inverses we denote as $A^{-1}$ and $\tilde{A}^{-1}$. Let $\mathcal{A}^{-1}:=(A^{-1})^{\otimes r} \otimes (\tilde{A}^{-1})^{\otimes (\tilde{N} - r)}$. On the one hand,
    \begin{equation}
    U_\pi \mathcal{A}^{-1} \ket{\psi} = U_\pi \left( 
    \begin{array}{c} 
		\begin{tikzpicture}[scale=.45,thick,baseline={([yshift=0ex]current bounding box.center)}]
            \draw (-0.5,-0.5) -- (-0.5,0) -- (0-0.15,0) -- (0-0.15,0.6);
            \draw (0+0.15,0.6) -- (0+0.15,0) -- (1-0.15,0) -- (1-0.15,0.6);
            \draw (1+0.15,0.6) -- (1+0.15,0) -- (1+0.5,0);
            \draw[dotted] (1+0.5,0) -- (3-0.5,0);
            \draw (3-0.5,0) -- (3-0.15,0) -- (3-0.15,0.6);
            \draw (3+0.15,0.6) -- (3+0.15,0) -- (3+0.5,0) -- (3+0.5,-0.5) -- (-0.5,-0.5);
		\end{tikzpicture}
    \end{array}  \right) =
    \begin{array}{c} 
		\begin{tikzpicture}[scale=.45,thick,baseline={([yshift=0ex]current bounding box.center)}]
            \draw (-1.5,-0.5) -- (-1.5,0) -- (-1-0.15,0) -- (-1-0.15,1.8);
            \draw (-1+0.15,1.8) -- (-1+0.15,0) -- (-1+1-0.15,0) -- (-1+1-0.15,0.5);
            \draw (-0.5,0) -- (0-0.15,0) -- (0-0.15,0.6);
            \draw (0+0.15,0.4) -- (0+0.15,0) -- (1-0.15,0) -- (1-0.15,0.5);
            \draw (1+0.15,0.6) -- (1+0.15,0) -- (1+0.5,0);
            \draw[dotted] (1.5,0) -- (3,0);
            \draw (3.5-0.5,0) -- (3.5-0.15,0) -- (3.5-0.15,0.4);
            \draw (3.5+0.15,0.6) -- (3.5+0.15,0) -- (4,0) -- (4,-0.5) -- (-1.5,-0.5);
            %
            % New permuted lines:
            \draw[thick] (0-0.15,0.6) to[out=90, in=-90] (3-0.15,1.8);
            \draw[thick] (0+0.15,0.4) to[out=90, in=-90] (3+0.15,1.6)--(3+0.15,1.8);
            \draw[thick] (1-0.15,0.5) to[out=90, in=-90] (0-0.15,1.8)--(-0.15,1.8);
            \draw[thick] (1+0.15,0.6) to[out=90, in=-90] (0+0.15,1.8);
            \draw[thick] (3.5-0.15,0.4) to[out=90, in=-90] (1.3-0.15,1.6) -- (1.3-0.15,1.8);
            \draw[thick] (3.5+0.15,0.6) to[out=90, in=-90] (1.3+0.15,1.8);
            %
            % Bipartition:
            \draw[densely dashed,draw=red] (2,-1.6)  -- (2,2.1);
            \draw (3.2,-1.2) node {\scriptsize $D_\mathds{1}^{2 \lfloor \frac{\tilde{N}}{2} \rfloor}$};
		\end{tikzpicture}
    \end{array} \label{eq:rank_counting_normal_proof_aux}.
    \end{equation}
    The Schmidt rank of this state across the half-chain bipartition $\{1, \dots, \lceil \frac{\tilde{N}}{2} \rceil \} \cup \{\lceil \frac{\tilde{N}}{2} \rceil + 1, \dots, \tilde{N}\}$ of the resulting state can be computed exactly, since it consists just of Bell pairs. Its value is $D_\mathds{1}^{2 \lfloor \frac{\tilde{N}}{2} \rfloor}$, which is equal to $D_\mathds{1}^{\tilde{N}}$ for even $\tilde{N}$, or $D_\mathds{1}^{\tilde{N}-1}$ for odd $\tilde{N}$.

    On the other hand, noting that $U_\pi$ and $\mathcal{A}^{-1}$ satisfy that $U_\pi \mathcal{A}^{-1} = \mathcal{A}_\pi^{-1} U_\pi$ for a product operator $\mathcal{A}_\pi^{-1}$ that is a product of $A^{-1}$ and $\tilde{A}^{-1}$ in a different ordering,
    \begin{equation}
    U_\pi \mathcal{A}^{-1} \ket{\psi} = \mathcal{A}_\pi^{-1} U_\pi \ket{\psi} = \begin{array}{c}
		\begin{tikzpicture}[scale=.5, baseline={([yshift=-3ex]current bounding box.center)}, thick, every node/.style={scale=0.9}]
            \begin{scope}[shift={(0,0)}]
        		\draw[shift={(0,0)},dotted] (0.5+0.3,0) -- (4+0.3,0);
                \MPSTensorLeft{(0,0)}{$A_\pi^{[1]}$}{amaranth}
                \MPSTensor{1.8,0}{$A_\pi^{[2]}$}{amaranth}
                \MPSTensorRight{4.5,0}{$A_\pi^{[\tilde{N}]}$}{amaranth}
                %\draw (-1,0) -- (-1,-0.8) -- (5.5,-0.8) -- (5.5,0);
        	\end{scope}
            \filldraw[shift={(0,1.5)},fill=purple] (-1/2-0.3,-1/2) -- (-1/2-0.3,1/2) -- (1/2+0.3,1/2) -- (1/2+0.3,-1/2) -- (-1/2-0.3,-1/2);
            \draw (-0.3, 2) -- (-0.3, 2.5);
            \draw(0.3, 2) -- (0.3, 2.5);
            \draw (0,1.5) node {\scriptsize $\mathcal{A}_{\pi}^{-1,[1]}$};
            %
            %\InverseMPS{(0,1.5)}{$\mathcal{A}_{\pi,1}^{-1,[1]}$}{purple}
            \filldraw[shift={(1.5+0.3,1.5)},fill=purple] (-1/2-0.3,-1/2) -- (-1/2-0.3,1/2) -- (1/2+0.3,1/2) -- (1/2+0.3,-1/2) -- (-1/2-0.3,-1/2);
            \draw (1.5, 2) -- (1.5, 2.5);
            \draw(2.1, 2) -- (2.1, 2.5);
            \draw (1.8,1.5) node {\scriptsize $\mathcal{A}_{\pi}^{-1,[2]}$};
            %
            %\InverseMPS{(1.5,1.5)}{$\mathcal{A}_{\pi,2}^{-1,[2]}$}{purple}
            \filldraw[shift={(4.5,1.5)},fill=purple] (-1/2-0.35,-1/2) -- (-1/2-0.35,1/2) -- (1/2+0.35,1/2) -- (1/2+0.35,-1/2) -- (-1/2-0.35,-1/2);
            \draw (4.2, 2) -- (4.2, 2.5);
            \draw(4.8, 2) -- (4.8, 2.5);
            \draw (4.5,1.5) node {\scriptsize $\mathcal{A}_{\pi}^{-1,[\tilde{N}]}$};
            \draw[densely dashed,draw=red] (3.15,-1.2)  -- (3.15,2.6);
            \draw (3.8,-0.9) node {\scriptsize $\leq D$};
        \end{tikzpicture}
    \end{array}, \label{eq:rank_proof_aux}
    \end{equation}
    where we used the MPS-up$_{0,D}$ property to substitute $U_\pi \ket{\psi}$ by an MPS of bond dimension $D_\pi \leq D$. 
    
    To be consistent with the previous expression, it is necessary that $D_\mathds{1}^{2 \lfloor \tilde{N}/2 \rfloor} \leq D$, which can only hold if \textit{(i)}~$D_\mathds{1} > 1$ but $\tilde{N}$ is not too large with respect to $D$, so that the equation is not violated, or \textit{(ii)} $D_\mathds{1} = 1$ and thus $\ket{\psi}$ is a product state. Situation \textit{(i)} is impossible for any $D_\mathds{1} > 1$ if $4^{\lfloor \tilde{N}/2 \rfloor} > D$, and the claim follows.
\end{proof}

\subsection{Exact MPS-up with nonnormal tensor} \label{app:proof_rankcounting-nonnormal}

Here we show how to remove the assumption of normality to extend the applicability of our results to all TI exact MPS-up. First, we give more details needed about the CF for nonnormal tensors. Any TI MPS $\ket{\psi}$ with tensor $B$ can be written in a block-diagonal form upon an appropriate gauge transformation,
\begin{equation*} %\label{eq:CF_def}
    B^i = \bigoplus_{j=1}^l \gamma_{j} B_j^i, 
\end{equation*}
where each $B_j$ has no nontrivial invariant subspace, and its associated CP map $\mathcal{E}_j$ has largest eigenvalue 1 \cite{Cirac_review_2021}. However, if the associated CP map $\mathcal{E}_j$ has more than one eigenvalue of magnitude 1, of the form $e^{i 2\pi q_i/p_i}$ for some $q_i, p_i \in \mathbb{Z}$, with $gcd(q_i,p_i) = 1$ and $p_i > 1$, then there are periodic subspaces that will cause the appearance of more invariant subspaces upon blocking.

Since the number of particles $N$ is fixed in our set-up and $\ket{\psi_N(B_i)} = 0$ if $p_i \not| N$ \cite{de_las_cuevas_irreducible_2017}, we will assume that $p_i \mid N$ for all blocks $C_i$. Let $p := lcm (\{p_j\})$, which satisfies $p \mid N$, and define the blocked tensor $A$ as $A^{i_1 \dots i_p} := B^{i_1} \dots B^{i_p}$. Then, $\ket{\psi} = \ket{\psi_{\bar{N}}(A)}$ with $\bar{N} := N/p$, and tensor $A$ no longer has any periodic subspaces. Now, $A$ can be written in the so-called \textit{canonical form} (CF) as 
\begin{equation} \label{eq:CF_def}
    A^i = \bigoplus_{j=1}^b \bigoplus_{q=1}^{r_j} \mu_{j,q} A_j^i, 
\end{equation}
where $|\mu_{j,q}| \leq 1$, $\forall j,q$, and the normal tensors $\{A_j\}_{j=1, \dots, b}$ form the so-called \textit{basis of normal tensors} (BNT) \cite{Cirac2017}, which has the property that each of its elements can be accessed separately by acting on the physical index after blocking at least $L_{BI} \leq 3(b-1)(L_0+1)$ sites together ($L_0$ is the maximum injectivity length over all blocks in the BNT) \cite{Perez-Garcia2007}. This is known as the \textit{block-injectivity} property, and it is equivalently expressed as the existence of operators $\mathbb{P}_i$ such that
\begin{equation} \label{eq:def_block_inj}
    \begin{array}{c} 
		\begin{tikzpicture}[scale=.4,thick,baseline={([yshift=0ex]current bounding box.center)}]
            \MPSTensor{0,0}{$a_i$}{purple}
            \PhysicalOperator{0,1.5}{$\mathbb{P}_i$}{yellow}
		\end{tikzpicture} = 
        \begin{tikzpicture}[scale=.4,thick,baseline={([yshift=0ex]current bounding box.center)}]
            \MPSTensor{0,0}{$a_i$}{purple}
		\end{tikzpicture} \ , \ \
        \begin{tikzpicture}[scale=.4,thick,baseline={([yshift=0ex]current bounding box.center)}]
            \MPSTensor{0,0}{$a_i$}{purple}
            \PhysicalOperator{0,1.5}{$\mathbb{P}_j$}{yellow}
		\end{tikzpicture}
    \end{array} \ = 0 , \ \forall i \neq j ,
\end{equation}
where the tensor $a_i^{k_1 \dots k_{L_{BI}}} := A_i^{k_1} \dots A_i^{k_{L_{BI}}}$. Therefore, the state can be written as
\begin{equation} \label{eq:state_nonnormal}
    \ket{\psi} = \frac{1}{c_N} \sum_{j=1}^b \alpha_j \ket{\psi_j}.
\end{equation}
for some normalization constant $c_N$, $\alpha_j := \sum_{j=1}^{r_j} (\mu_{j,q})^{\bar{N}}$, and $\ket{\psi_j} := \ket{\psi_{\bar{N}}(A_j)}$. The normal case studied before is recovered when $b = 1$ and there is just one non-zero coefficient, $\mu_{1,1} = 1$.

In the following statement, we show that any TI MPS with the exact MPS-up$_{D}$ property has a GHZ-type entanglement structure if the number of particles $N$ is large enough, since it can be written as a superposition of $b$ linearly independent product states. This reflects the fact that non-normal MPS-up should consist of $b$ independent physical sectors due to the block-injectivity property, with little entanglement.
%\exactMPSUPthm*
\begin{mythm}{1}[Exact MPS-up] \label{thm:exact_MPS-up_thm}
    Let $\ket{\psi_N(A)}$ be a TI MPS with the exact MPS-up$_{D}$ property on $N$ sites, with $N > pL_{BI}(\log_2 D + 1)$. Then, $\ket{\psi} = \sum_{i=1}^b \beta_i \ket{\phi_i}^{\otimes N}$, where $b$ denotes the number of elements in the BNT of tensor $A$.
\end{mythm}
\begin{proof}
    This can be shown analogously to Proposition \ref{prop:rank_counting_normal}. Given a TI MPS with tensor $B$, we start by blocking every $p$ sites together to remove periodicities. This results in a TI MPS with blocked tensor $A$ on $\bar{N} = N/p \in \mathbb{Z}$ sites, with canonical form as in Eq. (\ref{eq:CF_def}) and BNT $\{A_1, \dots, A_b\}$.
    
    Let $\tilde{N} := \lfloor \bar{N}/L_{BI} \rfloor$. Define $\pi \in \mathcal{S}_{\tilde{N}}$ as the permutation in Eq. (\ref{eq:def_permutation_rank_counting}), and $r$ as the unique integer such that $\bar{N} = L_{BI} \tilde{N} + r$ with $0 \leq r < L_{BI}$. We block tensors like we did in Proposition \ref{prop:rank_counting_normal}, ending up with $r$ blocks of $L_{BI}+1$ sites with tensor $a$ (where $a^{i_1 \dots i_{L_{BI}+1}} := A^{i_1} \dots A^{i_{L_{BI}+1}}$) and $\tilde{N}-r$ blocks of $L_{BI}$ sites with tensor $\tilde{a}$ (where $\tilde{a}^{i_1 \dots i_{L_{BI}}} := A^{i_1} \dots A^{i_{L_{BI}}}$). Note that the number of elements in the BNT and the bond dimensions $D_{i, \mathds{1}}$ and $\tilde{D}_{i, \mathds{1}}$ of each element remain stable upon blocking.
    
    Let $\mathcal{A}_i^{-1} := (a_i^{-1} \mathbb{P}_i)^{\otimes r} \otimes (\tilde{a}_i^{-1} \tilde{\mathbb{P}}_i)^{\otimes (\tilde{N}-r)}$. Then, 
\begin{align*}
    U_\pi \mathcal{A}^{-1}_i \ket{\psi} &= U_\pi \left( \begin{array}{c}
		\begin{tikzpicture}[scale=.5, baseline={([yshift=-3ex]current bounding box.center)}, thick]
            \FullMPS{(0,0)}{\small $a$}{purple}
            \begin{scope}[shift={(4.5,0)}]
        		\draw (-1,0) -- (1,0);
        		\draw (0,1) -- (0,0);
        		\filldraw[fill=purple] (-1/2,-1/2) -- (-1/2,1/2) -- (1/2,1/2) -- (1/2,-1/2) -- (-1/2,-1/2);
        		\draw (0,0) node {\small $\tilde{a}$};
        	\end{scope}
            \PhysicalOperator{(0,1.5)}{$\mathbb{P}_i$}{yellow};
            \PhysicalOperator{(1.5,1.5)}{$\mathbb{P}_i$}{yellow};
            \PhysicalOperator{(4.5,1.5)}{$\tilde{\mathbb{P}}_i$}{yellow};
            \InverseMPS{(0,3)}{$a_i^{-1}$}{purple};
            \InverseMPS{(1.5,3)}{$a_i^{-1}$}{purple};
            \InverseMPS{(4.5,3)}{$\tilde{a}_i^{-1}$}{purple};
        \end{tikzpicture}
    \end{array} \right) \\
    &= U_\pi \left( {\small \frac{\alpha_{i,N}}{c_N}} \begin{array}{c}
		\begin{tikzpicture}[scale=.5, baseline={([yshift=-3ex]current bounding box.center)}, thick]
            \FullMPS{(0,0)}{$a_i$}{purple}
            \begin{scope}[shift={(4.5,0)}]
        		\draw (-1,0) -- (1,0);
        		\draw (0,1) -- (0,0);
        		\filldraw[fill=purple] (-1/2,-1/2) -- (-1/2,1/2) -- (1/2,1/2) -- (1/2,-1/2) -- (-1/2,-1/2);
        		\draw (0,0) node {\scriptsize $\tilde{a}_i$};
        	\end{scope}
            \InverseMPS{(0,1.5)}{$a_i^{-1}$}{purple};
            \InverseMPS{(1.5,1.5)}{$a_i^{-1}$}{purple};
            \InverseMPS{(4.5,1.5)}{$\tilde{a}_i^{-1}$}{purple};
        \end{tikzpicture}
    \end{array} + 0\right) \\
    &= \frac{\alpha_{i,N}}{c_N}
    \begin{array}{c} 
		\begin{tikzpicture}[scale=.45,thick,baseline={([yshift=-0.2ex]current bounding box.center)}]
            \draw (-1.5,-0.5) -- (-1.5,0) -- (-1-0.15,0) -- (-1-0.15,1.8);
            \draw (-1+0.15,1.8) -- (-1+0.15,0) -- (-1+1-0.15,0) -- (-1+1-0.15,0.5);
            \draw (-0.5,0) -- (0-0.15,0) -- (0-0.15,0.6);
            \draw (0+0.15,0.4) -- (0+0.15,0) -- (1-0.15,0) -- (1-0.15,0.5);
            \draw (1+0.15,0.6) -- (1+0.15,0) -- (1+0.5,0);
            \draw[dotted] (1.5,0) -- (3,0);
            \draw (3.5-0.5,0) -- (3.5-0.15,0) -- (3.5-0.15,0.4);
            \draw (3.5+0.15,0.6) -- (3.5+0.15,0) -- (4,0) -- (4,-0.5) -- (-1.5,-0.5);
            % New permuted lines:
            \draw[thick] (0-0.15,0.6) to[out=90, in=-90] (3-0.15,1.8);
            \draw[thick] (0+0.15,0.4) to[out=90, in=-90] (3+0.15,1.6)--(3+0.15,1.8);
            \draw[thick] (1-0.15,0.5) to[out=90, in=-90] (0-0.15,1.8)--(-0.15,1.8);
            \draw[thick] (1+0.15,0.6) to[out=90, in=-90] (0+0.15,1.8);
            \draw[thick] (3.5-0.15,0.4) to[out=90, in=-90] (1.3-0.15,1.6) -- (1.3-0.15,1.8);
            \draw[thick] (3.5+0.15,0.6) to[out=90, in=-90] (1.3+0.15,1.8);
            % Bipartition:
            \draw[densely dashed,draw=red] (2,-1.6)  -- (2,2.1);
            \draw (3.2,-1.2) node {\scriptsize $D_{i,\mathds{1}}^{2 \lfloor \frac{\tilde{N}}{2} \rfloor}$};
		\end{tikzpicture}
    \end{array}.
\end{align*}
The Schmidt rank across the half-chain bipartition $\{1, \dots, \lceil \frac{\tilde{N}}{2} \rceil \} \cup \{\lceil \frac{\tilde{N}}{2} \rceil + 1, \dots, \tilde{N}\}$ of the resulting state can be computed exactly again here, which is $D_{i,\mathds{1}}^{2 \lfloor \frac{\tilde{N}}{2} \rfloor}$.

On the other hand, since $U_\pi$ and $\mathcal{A}_i^{-1}$ satisfy that $U_\pi \mathcal{A}_i^{-1} = \mathcal{A}_{i,\pi}^{-1} U_\pi$ for a product operator $\mathcal{A}_{i,\pi}^{-1}$ consisting of the product of $a_i^{-1} \mathbb{P}_i$ and $\tilde{a}_i^{-1} \tilde{\mathbb{P}}_i$ in a different ordering, we can use the MPS-up$_{0,D}$ property as was done in Eq. (\ref{eq:rank_proof_aux}) to obtain that the Schmidt rank should be no larger than $D$.

Therefore, it is necessary that $D_{i,\mathds{1}}^{2 \lfloor \frac{\tilde{N}}{2} \rfloor} \leq D$, which can only hold if \textit{(i)} $D_{i,\mathds{1}} > 1$ but $\tilde{N}$ is not too large with respect to $D$, so that the equation is not violated, or \textit{(ii)} $D_{i,\mathds{1}} = 1$ and thus $\ket{\psi_i}$ is a product state. The same conclusion holds for each $i$-th block of the BNT, and the claim follows by Eq. (\ref{eq:state_nonnormal}) and by noting that situation \textit{(i)} is impossible for any $D_{\mathds{1}} > 1$ if $4^{\lfloor \tilde{N}/2 \rfloor} > D$.

By leveraging the translational invariance of the MPS, we can further conclude that each of the terms $\ket{\psi_j}$ in Eq. \eqref{eq:state_nonnormal} takes the form $|\phi_i\rangle^{\otimes N}$.
\end{proof}

\section{Proofs for approximate MPS-under-permutations}
\label{app:proof_puritygeneral}

\subsection{Technical lemma for normal approximate MPS-up} \label{app:proof_technical-lemma}

A technical lemma that we use in the proofs for the approximate MPS-up property is the following. 
\begin{lemma} \label{lemma:trace_distance}
    Given two states $\ket{\psi_1}, \ket{\psi_2}$ with density matrices $\rho_1, \rho_2$, and reduced density matrices over subset $S$ denoted as $\rho_1^s = \Tr_{s^c}[\rho_1], \rho_2^s:= \Tr_{s^c}[\rho_2]$, then
    \begin{equation*}
         |\Tr[(\rho_1^s)^2] - \Tr[(\rho_2^s)^2]|
         \leq 4 \| \ket{\psi_1} - \ket{\psi_2} \|.
    \end{equation*}
\end{lemma}
\begin{proof}
    Using the fact that $|\Tr[AB]| \leq \| A \|_\infty \| B \|_1$,
    \begin{align}
       |&\Tr[(\rho_1^s)^2] - \Tr[(\rho_2^s)^2]| \nonumber \\
       &= \Tr[\rho_1^s (\rho_1^s - \rho_2^s)] + \Tr[\rho_2^s (\rho_1^s - \rho_2^s)]  \nonumber \\
       &\leq \left( \| \rho_1^s \|_\infty + \| \rho_2^s \|_\infty \right) \| \rho_1^s - \rho_2^s \|_1 \nonumber \\
       &\leq 2 \| \rho_1 - \rho_2 \|_1, \label{eq:lemma_aux2}
    \end{align}
    where we used $\| \rho_i^s\|_\infty \leq 1$ and the monotonicity of the trace distance under the partial trace in the last line. On the other hand, due to the inequality between trace distance and fidelity (see Theorem 9.3.1 in \cite{Wilde2013}), we have
    \begin{align}
        \|\rho_1 - \rho_2\|_1 &\leq 2 \sqrt{1 - |\braket{\psi_1}{\psi_2}|^2} \nonumber \\
        &\leq 2 \sqrt{1 - \text{Re}(\braket{\psi_1}{\psi_2}) ^2}. \label{eq:lemma_aux1}
    \end{align}
    Noting that $\| \ket{\psi_1} - \ket{\psi_2} \|^2 = 2(1 - \text{Re}(\braket{\psi_1}{\psi_2}))$,
    \begin{align*}
        \text{Re}(\braket{\psi_1}{\psi_2})^2 &= 1 - \| \ket{\psi_1} - \ket{\psi_2}\|^2 + \frac{\|\ket{\psi_1} - \ket{\psi_2}\|^4}{4} \\
        &\geq 1 - \| \ket{\psi_1} - \ket{\psi_2} \|^2, 
    \end{align*}
    and hence from Eq. (\ref{eq:lemma_aux1}),
    \begin{equation*}
        \| \rho_1 - \rho_2 \|_1 \leq 
        2 \| \ket{\psi_1} - \ket{\psi_2} \|.
    \end{equation*}
    Plugging this into Eq. (\ref{eq:lemma_aux2}) completes the proof of the lemma. 
\end{proof}

\subsection{Approximate MPS-up with nonnormal tensor} \label{app:proof_puritynonnormal}

Here we show how to tackle the general case of families of TI MPS $\{\ket{\psi_N(A)}\}_N$ for any tensor $A$, with the approximate MPS-up$_{\varepsilon, D}$ property with $\varepsilon \geq 0$ and for all $N > N_0$ for some constant $N_0$. In order to do so, we introduce a notion of angle $\theta_i \in [0, \frac{\pi}{2}]$ quantifying the distinguishability of each block $A_i$ in the BNT with respect to the other blocks. Denoting as $V_i$ the physical subspaces associated to element $A_i$ of the BNT, where 
\begin{equation*}
    V_i := \left\{ 
    \begin{tikzpicture}[scale=.4,thick,baseline={([yshift=-0.7ex]current bounding box.center)}]
        \simplematrix{0,0}{$X$}{yellow}
        \MPSTensor{1.5,0}{$A_i$}{purple}
        \draw (-1,0) -- (-1,-0.8) -- (2.5,-0.8) -- (2.5,0);
	\end{tikzpicture} \ \bigg| \
    \begin{tikzpicture}[scale=.4,thick,baseline={([yshift=-0.5ex]current bounding box.center)}]
        \simplematrix{0,0}{$X$}{yellow}
    \end{tikzpicture} \in \mathcal{M}_{D_i}(\mathbb{C})
    \right\} \subseteq \mathbb{C}^d,
\end{equation*}
the quantity $\theta_i$ is defined as the angle between the physical subspace $V_i$ and its complement $V_i^c$ \cite{Deutsch_angle_1995}, where $V_i^c := V_1 + \dots + \hat{V}_i + \dots + V_m$, and the hat indicates that the term is omitted in the sum. That is,
\begin{equation*}
    \cos \theta_i := \sup \{ | \braket{x}{y} | \mid x \in V_i, \ y \in V_i^c, \ \| x \|, \| y \| \leq 1 \}.
\end{equation*} 
The block-injectivity property is equivalent to $V_i \cap V_i^c = \{0\}$ for all $i$ \cite{Perez-Garcia2007}, implying that $\cos\theta_i < 1$. The operator $\mathbb{P}_i$ defined in Eq. (\ref{eq:def_block_inj}) is the unique projection satisfying Im$(\mathbb{P}_i) = V_i$ and Ker$(\mathbb{P}_i) = V_i^c$, and its operator norm is related to the angle $\theta_i$ as $\| \mathbb{P}_i \|_\text{op} = \csc\theta_i \geq 1$, with $\| \mathbb{P}_i \|_\text{op} = 1$ if and only if $V_i$ and $V_i^c$ are orthogonal subspaces \cite{Buckholtz_projector_1999}.

If we assume that every $L$ sites are blocked together, and let $V_{i,L}$ denote the physical subspace of the blocked $i$-th element in the BNT, with $\theta_{i,L}$ being its corresponding angle, then $\cos\theta_{i,L} = \text{sup}\{ \Tr[(X_1 \otimes X_2) \mathbb{E}^L_{ij}] \mid \forall X_1, X_2, j \neq i \} = O(\max_{j\neq i} |\lambda_{ij}|^L)$, where $\lambda_{ij}$ denotes the maximum eigenvalue of $\mathbb{E}_{ij}$ ($|\lambda_{ij}| < 1$ by Lemma A.2 in \cite{Cirac2017}). Therefore, $\sin\theta_{i,L} = 1+O(\xi_i^{2L})$, where $\xi_i := \max_j |\lambda_{ij}|$, and the blocks become orthogonal to each other as $L \to \infty$. 

We restate Theorem \ref{thm:approx_MPS-up_thm} below for convenience, and proceed to prove it.

\begin{mythm}{2}[Approximate MPS-up]
    Let $\{\ket{\psi_N(A)}\}_N$ be a family of TI MPS with the approximate MPS-up$_{\varepsilon_N, D_N}$ property for all $N$ larger than some $N_0$, where $D_N = O(\text{poly}(N))$. 
    Let $b$ be the number of elements in the BNT of tensor $A$.
    Then, if there exists a positive sequence $(g_N)$ with $g_N = \Omega(1/\text{poly}(N))$, such that either
    \begin{enumerate}[(a)]
        \item $b = 1$ and $0 \leq \varepsilon_N < \frac{1}{4D_N}-g_N$, or
        \item $b > 1$ and $0 \leq \varepsilon_N <  \left( \frac{1}{4D_N} - g_N \right) \frac{\min_i |\alpha_{i}|}{2(\sum_{i} |\alpha_{i}|^2)^{\frac{1}{2}}}$, 
    \end{enumerate}
    where $\alpha_i$ denote the coefficients weighting each element of the BNT of $A$ according to Eq. \eqref{eq:nonnormal_MPS_intro}, then $\ket{\psi_N(A)} = \sum_{i=1}^b \beta_i \ket{\phi_i}^{\otimes N}$.
\end{mythm}
\begin{proof}
Part (a) of the statement has already been shown in Proposition \ref{prop:purity_normal} in the main text. Here we proceed to prove part (b). Assume that every $L \geq pL_{BI}$ sites have already been blocked together, such that tensor $A$ is block-injective. Let $\mathbb{P}_{i,L}$ be the physical projector accessing the $i$-th element of the BNT of the blocked tensor. For simplicity, consider $N = kL$ for $k \in \mathbb{N}$. 

Due to the MPS-up$_{\varepsilon_N, D_N}$ property, we know that for each $N$ and each permutation $\pi$, there exists a (normalised) MPS state $\ket{\psi_\pi}$ that approximates $U_\pi \ket{\psi}$ with bond dimension at most $D_N$, so that
\begin{equation*}
\| U_\pi \ket{\psi} - \ket{\psi_\pi} \| \leq \varepsilon. 
\end{equation*}
Then, noting that $U_\pi$ and $\mathbb{P}_{i,L}^{\otimes (N/L)}$ commute, we have
\begin{align}
    &\left\| U_\pi \ket{\psi_i} - \frac{c_N}{\alpha_i} \mathbb{P}_{i,L}^{\otimes (N/L)} \ket{\psi_\pi} \right\| \nonumber \\
    &\quad = \frac{c_N}{|\alpha_i|} \| U_\pi \mathbb{P}_{i,L}^{\otimes (N/L)} \ket{\psi} - \mathbb{P}_{i,L}^{\otimes (N/L)} \ket{\psi_\pi} \| \nonumber \\
    &\quad = \frac{c_N}{|\alpha_i|} \| \mathbb{P}_{i,L}^{\otimes (N/L)} ( U_\pi \ket{\psi} - \ket{\psi_\pi} ) \| \nonumber \\
    &\quad \leq \frac{c_N}{|\alpha_i|} \varepsilon \|\mathbb{P}_{i,L} \|^{N/L}_\text{op}
    \label{eq:purity_nonnormal_2},
\end{align}
where $\|\mathbb{P}_{i,L} \|_\text{op} = \csc \theta_{i,L}$, and $c_N$ is the normalization constant appearing in Eq. \eqref{eq:state_nonnormal}. Now, in order to obtain the MPS-up property individually for each block of the BNT, we show that the normalized state associated to $e^{-i\beta_i} \mathbb{P}_{i,L}^{\otimes (N/L)} \ket{\psi_\pi}$ gives the desired approximation to $U_\pi \ket{\psi_i}$, where $\beta_i$ is the angle appearing in $\alpha_i = |\alpha_i| e^{i\beta_i}$. Then,
\begin{align*}
    &\left\| \frac{U_\pi \ket{\psi_i}}{\| \ket{\psi_i}\|} - e^{-i\beta_i} \frac{\mathbb{P}_{i,L}^{\otimes (N/L)}  \ket{\psi_\pi}}{\| \mathbb{P}_{i,L}^{\otimes (N/L)}  \ket{\psi_\pi} \|} \right\| \\
    &\quad \leq 
    \frac{1}{\|\ket{\psi_i}\|} \left\| U_\pi \ket{\psi_i} - \frac{c_N}{\alpha_i} \mathbb{P}_{i,L}^{\otimes (N/L)} \ket{\psi_\pi} \right\| \\
    &\quad \quad +
    \left| \frac{c_N}{|\alpha_i|} \frac{\| \mathbb{P}_{i,L}^{\otimes (N/L)} \ket{\psi_\pi} \|}{\| \ket{\psi_i}\|} - 1 \right|
    \\
    &\quad \leq \frac{2 \varepsilon}{|\alpha_i|} \frac{c_N}{\| \ket{\psi_i} \|} \| \mathbb{P}_{i,L}\|_\text{op}^{N/L}
\end{align*}
where we used the triangle inequality, and Eq. (\ref{eq:purity_nonnormal_2}) together with the fact that $| \| x \| - \| y \| | \leq \| x-y \|$ in order to upper bound the second term. Therefore, for $$|\tilde{\psi}_\pi\rangle := e^{-i\beta_i} \frac{\mathbb{P}_{i,L}^{\otimes (N/L)}  \ket{\psi_\pi}}{ \| \mathbb{P}_{i,L}^{\otimes (N/L)}  \ket{\psi_\pi}\|},$$ which is an MPS of bond dimension $\leq D_N$, we have that the state $\ket{\psi_i}$ generated by the $i$-th element of the BNT has the MPS-up$_{\tilde{\varepsilon}_{i,N}, D_N}$ property,
\begin{equation*} 
\left\| \frac{U_\pi \ket{\psi_i}}{\| \ket{\psi_i}\|} - |\tilde{\psi}_\pi\rangle \right\| \leq \underbrace{\frac{2\varepsilon_N}{|\alpha_{i}|}  \frac{c_N}{\Tr[\mathbb{E}_{A_i}^N]} \| \mathbb{P}_{i,L} \|_\text{op}^{N/L} }_{=: \tilde{\varepsilon}_{i,N}}.
\end{equation*}

If it holds that $\tilde{\varepsilon}_{i,N} < \frac{1}{4 D_N} - g_N$, for some positive sequence $(g_N)$, with $g_N = \Omega(1/\text{poly}(N))$, for all $i$ and for all $N$ sufficiently large, or equivalently in terms of the original parameter $\varepsilon_N$, if
\begin{equation} \label{eq:upper_bound_first_version}
    \varepsilon_N <  \left( \frac{1}{4D_N} - g_N \right) \frac{\min_i |\alpha_{i}|  (\sin\theta_{i,L})^{N/L} }{2(\sum_{i} |\alpha_{i}|^2)^{\frac{1}{2}}},
\end{equation}
then we can use Proposition \ref{prop:purity_normal} to conclude that each $A_i$ in the BNT has bond dimension one, and thus $\ket{\psi_i}$ are product states for all $i$. Note that $\alpha_i$ could depend on $N$. To obtain Eq. (\ref{eq:upper_bound_first_version}), we used $\Tr[\mathbb{E}^N_{A_i}] \to 1$ as $N \to \infty$ and $c_N \to (\sum_{j=1}^b |\alpha_{j,N}|^2)^{1/2}$, because
\begin{align*}
    c_N &= \left(\sum_{j,j' = 1}^b \alpha_{j'}^* \alpha_j \braket{\psi_{j'}}{\psi_j}\right)^\frac{1}{2} \\
    &= \left(\sum_{j = 1}^b |\alpha_j|^2 \|\ket{\psi_j}\|^2 + \sum_{j\neq j'} \alpha_{j'}^* \alpha_j \Tr[\mathbb{E}_{jj'}^N]\right)^\frac{1}{2},
\end{align*}
where $\| \ket{\psi_j} \| \to 1$ as $N \to \infty$, and the transfer matrices $\mathbb{E}_{jj'} := \sum_i A_j^i \otimes (A_{j'}^i)^*$ have leading eigenvalue strictly smaller than 1 for $j \neq j'$ (lemma A.2 in \cite{Cirac2017}).

The condition in Eq. (\ref{eq:upper_bound_first_version}) might appear overly restrictive due to the exponentially decaying term. However, this is not actually the case, due to the fact that the blocks become orthogonal in the limit $L \to \infty$, as pointed out before.

To make this more explicit, take $L = O(\log N)$. Then, we have that $(\sin\theta_{i,L})^{N/L} = (1+O(\xi_i^{2L}))^{N/L} = (1 + O(1/N) )^{N/\log N} \xrightarrow{N \to \infty} 1$. Thus, even if $\sin \theta_{i,L} < 1$, we can rewrite the sufficient condition on $\varepsilon_N$ for $\ket{\psi_i}$ to be a product state as
\begin{equation}
    \varepsilon_N <  \left( \frac{1}{4D_N} - g_N \right) \frac{\min_i |\alpha_{i}|}{2(\sum_{i} |\alpha_{i}|^2)^{\frac{1}{2}}},
\end{equation}
for sufficiently large $N$. 
\end{proof}

\vspace{1mm}

\section{Proofs for non-TI MPS-under-permutations} \label{app:non-TI-MPSup}

In this appendix we establish the following proposition from the main text:
\begin{myprop}{4}[Approximate non-TI MPS-up]
    A non-TI MPS with exponentially decaying correlations and the MPS-up$_{\varepsilon, D}$ property with $\varepsilon < \frac{1}{4D}$, on a sufficiently large number of particles, can be written as a product of almost all sites or blocks of a few sites.
\end{myprop}
The proof proceeds in two steps, each detailed in the following two subsections. In section \ref{app:proof_nonTI-MPS-properties} we show that any such non-TI MPS is ergodic, meaning thats its associated sequence of quantum channels converges exponentially fast to a rank-one replacement channel. Then, in section \ref{app:proof_nonTI_genericMPS}, we use ergodicity together with the MPS-up property to conclude that the state must factorize into a product of almost all sites or blocks of a few sites.

\subsection{Non-TI MPS with exponentially decaying correlations are ergodic} \label{app:proof_nonTI-MPS-properties}

A state $\psi$ on $N$ sites has \textit{exponentially decaying correlations} \cite{brandao_2015_exp_decay_corr} if, for any given regions $X, Y$ separated by more than $l$ sites, there exist $l_0 \in \mathbb{N}$ and $\xi \in \mathbb{C}$ such that
\begin{align} 
    \text{Corr}(X:Y) &\equiv \max_{\substack{O,O' \\ \|O\|_{\text{op}}\leq 1 \\ \|O'\|_{\text{op}}\leq 1}} |\Tr[(O \otimes O')(\rho_{XY} - \rho_X \otimes \rho_Y)]| \nonumber \\
    &\leq Ke^{-l/\xi} \label{eq:exp_decay_corr}
\end{align}
for all $l \geq l_0$ and some constant $K$. Here $\rho_X$ denotes the reduced density matrix of $\ket{\psi}$ on region $X$.

Beyond it physical relevance and prevalence, exponentially decaying correlations are also \textit{generic}, in the sense that any random non-uniform MPS, chosen from the ensemble of MPS that can be sequentially prepared by a quantum computer, satisfies Eq. \eqref{eq:exp_decay_corr} \cite{haag_2023_typical_corr_length} and is in left-canonical form by construction. This means that there exist diagonal positive matrices $\Lambda_{[n]}$ with $\Tr[\Lambda_{[n]}] = 1$ such that
\begin{equation}\label{eq:decay_2}
    \begin{tikzpicture}[scale=.45, baseline={([yshift=-1ex]current bounding box.center)}, thick]
        \begin{scope}[shift={(0,0)}]
            \MPSTensordownrecthoriz{0,1.5}{$A_{[n]}$}{purple}{1.1}
            \MPSTensorrecthoriz{0,0}{$A_{[n]}^*$}{purple}{1.1}
        \end{scope}
        \draw (-1,0) -- (-1,1.5);
    \end{tikzpicture}
    \ = \ \begin{tikzpicture}[scale=.45, baseline={([yshift=-1ex]current bounding box.center)}, thick]
        \draw (0,0) -- (-1,0) -- (-1,1.5) -- (0,1.5);
    \end{tikzpicture}
    \quad
    \text{and}
    \quad 
    \begin{tikzpicture}[scale=.45, baseline={([yshift=-1ex]current bounding box.center)}, thick]
        \begin{scope}[shift={(0,0)}]
            \MPSTensordownrecthoriz{0,1.5}{$A_{[n]}$}{purple}{1.1}
            \MPSTensorrecthoriz{0,0}{$A_{[n]}^*$}{purple}{1.1}
        \end{scope}
        \draw (1,0) -- (1.5,0) -- (1.5,1.5) -- (1,1.5);
        \begin{scope}[shift={(1.5,0.75)}]
    		\filldraw[fill=yellow] (-1/2-0.1,-1/2) -- (-1/2-0.1,1/2) -- (1/2+0.1,1/2) -- (1/2+0.1,-1/2) -- (-1/2-0.1,-1/2);
    		\draw (0,0) node {\scriptsize $\Lambda_{[n]}$};
    	\end{scope}
    \end{tikzpicture}
    \ = \ \begin{tikzpicture}[scale=.45, baseline={([yshift=-1ex]current bounding box.center)}, thick]
        \draw (-1.5,0) -- (0,0) -- (0,1.5) -- (-1.5,1.5);
        \begin{scope}[shift={(0,0.75)}]
    		\filldraw[fill=yellow] (-1/2-0.4,-1/2) -- (-1/2-0.4,1/2) -- (1/2+0.4,1/2) -- (1/2+0.4,-1/2) -- (-1/2-0.4,-1/2);
    		\draw (0,0) node {\scriptsize $\Lambda_{[n-1]}$};
    	\end{scope}
    \end{tikzpicture}
\end{equation}
By blocking few sites together, one may moreover assume injectivity of each tensor without loss of generality, as explained in the main text. Finally, to avoid pathological ``near-singular'' cases, in which a tiny perturbation could break injectivity and induce long range correlations, we shall assume that each tensor is \textit{well-conditioned}, i.e. $\|A_{[i]}^{-1}\|_{\rm op} = O(1)$.

Under these generic conditions, we now prove that the MPS is \textit{ergodic}, meaning that its associated sequence of quantum channels converges exponentially fast to a replacement rank-one channel \cite{movassagh_2021_ergodic, movassagh_ergodic_2022} or, equivalently, if for any $i, j \in \mathbb{N}$ ($i \leq j$), there exist $\xi \geq 0$ and positive density operators $\Lambda_{[x]}$ such that
\begin{equation} \label{eq:condition_non-TI_purity_proof}
    \| \mathbb{E}_{i,j} - \Lambda_{[i-1]} \Tr[\cdot] \| \leq C e^{-|j-i|/\xi},
\end{equation}
for some constant $C = O(1)$, and $\mathbb{E}_{i,j}$ defined as
\begin{equation*}
    \begin{tikzpicture}[scale=.45, baseline={([yshift=-0.5ex]current bounding box.center)}, thick]
        \ETensor{0,0}{$\mathbb{E}_{i,j}$}{purple}{0.8}{0.7}
    \end{tikzpicture}
    := 
    \begin{tikzpicture}[scale=.45, baseline={([yshift=-0.5ex]current bounding box.center)}, thick]
        \ETensor{0,0}{$\mathbb{E}_{[i]}$}{purple}{0.8}{0.7}
        \ETensor{2.2,0}{$\mathbb{E}_{[i+1]}$}{purple}{0.9}{0.7}
        \node at (4.3,0.35) {\small $\dots$};
        \ETensor{6.4,0}{$\mathbb{E}_{[j-1]}$}{purple}{0.9}{0.7}
    \end{tikzpicture},
\end{equation*}
where $\mathbb{E}_{[x]}$ denotes the transfer matrix of the MPS tensor at the physical site $x$.

\begin{lemma} \label{lemma:nonTI-MPS-generic-implies-ergodic}
    An injective non-TI MPS on $N$ sites in left-canonical form, with exponentially decaying correlations and well-conditioned tensors, i.e. $\|A_{[i]}^{-1}\|_{\rm op} = O(1)$, is necessarily ergodic.
\end{lemma}
\begin{proof}
    Let $X, Y$ be the single-site regions $X = \{i-1\}$, $Y = \{j\}$, and for every $\ket{\alpha} \in \mathbb{C}^{D_{[i]}^2}$, $\ket{\beta} \in \mathbb{C}^{D_{[j]}^2}$ with $\| \ket{\alpha} \|_2 = \| \ket{\beta} \|_2 = 1$, define $O_\alpha, O'_\beta$ as
    \begin{equation*}
        O_\alpha \equiv 
        \begin{tikzpicture}[scale=.45, baseline={([yshift=-0.5ex]current bounding box.center)}, thick]
            \draw (0,-1) -- (0,1);
            \draw (1.6,-1) -- (1.6,1);
            \begin{scope}[shift={(0,0)}]
        		\filldraw[fill=yellow] (-1/2-0.4,-1/2) -- (-1/2-0.4,1/2) -- (1/2+0.4,1/2) -- (1/2+0.4,-1/2) -- (-1/2-0.4,-1/2);
        		\draw (0,0) node {\scriptsize $\Lambda_{[i-1]}$};
        	\end{scope}
            \begin{scope}[shift={(1.6,0)}]
        		\filldraw[fill=amaranth] (-1/2,-1/2) -- (-1/2,1/2) -- (1/2,1/2) -- (1/2,-1/2) -- (-1/2,-1/2);
        		\draw (0,0) node {\scriptsize $\alpha$};
        	\end{scope}
            % Ahora hacemos el tensor de arriba:
            \begin{scope}[shift={(0,1.5)}]
        		\filldraw[fill=purple] (-0.5,-1/2) -- (-0.5,1/2) -- (2.1,1/2) -- (2.1,-1/2) -- (-0.5,-1/2);
        	\draw (0.8,0) node {\scriptsize $A_{[i-1]}^{-1}$};
        	\end{scope}
            \draw (0.8,2) -- (0.8,2.5);
            % Ahora hacemos el tensor de abajo:
            \begin{scope}[shift={(0,-1.5)}]
        		\filldraw[fill=purple] (-0.5,-1/2) -- (-0.5,1/2) -- (2.1,1/2) -- (2.1,-1/2) -- (-0.5,-1/2);
        	\draw (0.8,0) node {\scriptsize $(A_{[i-1]}^{-1})^\ast$};
        	\end{scope}
            \draw (0.8,-2) -- (0.8,-2.5);
        \end{tikzpicture}
        , \quad O'_\beta \equiv \
        \begin{tikzpicture}[scale=.45, baseline={([yshift=-0.5ex]current bounding box.center)}, thick]
            \draw (0.2,-1) -- (0.2,1);
            \draw (1.4,-1) -- (1.4,1);
            \begin{scope}[shift={(0.2,0)}]
        		\filldraw[fill=amaranth] (-1/2,-1/2) -- (-1/2,1/2) -- (1/2,1/2) -- (1/2,-1/2) -- (-1/2,-1/2);
        		\draw (0,0) node {\scriptsize $\beta$};
        	\end{scope}
            % Ahora hacemos el tensor de arriba:
            \begin{scope}[shift={(0.8,1.5)}]
        		\filldraw[fill=purple] (-1/2-0.6,-1/2) -- (-1/2-0.6,1/2) -- (1/2+0.6,1/2) -- (1/2+0.6,-1/2) -- (-1/2-0.6,-1/2);
        	    \draw (0,0) node {\scriptsize $A_{[j]}^{-1}$};
        	\end{scope}
            \draw (0.8,2) -- (0.8,2.5);
            % Ahora hacemos el tensor de abajo:
            \begin{scope}[shift={(0.8,-1.5)}]
        		\filldraw[fill=purple] (-1/2-0.6,-1/2) -- (-1/2-0.6,1/2) -- (1/2+0.6,1/2) -- (1/2+0.6,-1/2) -- (-1/2-0.6,-1/2);
        	    \draw (0,0) node {\scriptsize $(A_{[j]}^{-1})^\ast$};
        	\end{scope}
            \draw (0.8,-2) -- (0.8,-2.5);
        \end{tikzpicture}
    \end{equation*}
    Using the fact that the MPS tensors are in left-canonical form, we can rewrite Eq. \eqref{eq:exp_decay_corr} as
    \begin{align}
         &K e^{-|i-j|/\xi} \nonumber \\
         &\geq \sup_{\substack{\alpha, \beta \\ \|\ket{\alpha}\|_2 = 1 \\ \| \ket{\beta} \|_2 = 1}}
         \frac{\left|\Tr[(O_\alpha \otimes O'_\beta) (\rho_{XY} - \rho_X \otimes \rho_Y)]\right|}{\| O_\alpha \|_{\text{op}} \| O'_\beta \|_{\text{op}}}
         \nonumber \\
         &= \sup_{\substack{\alpha, \beta \\ \|\ket{\alpha}\|_2 = 1 \\ \| \ket{\beta} \|_2 = 1}}
         \frac{1}{\| O_\alpha \|_{\text{op}} \| O'_\beta \|_{\text{op}}} \left| \
         \begin{tikzpicture}[scale=.45, baseline={([yshift=-0.5ex]current bounding box.center)}, thick]
            % Hacemos las lineas:
            \draw (-1.5,0) -- (-1.5,3) -- (7,3) -- (7,0) -- (-1.5,0);
            \draw (0,0) -- (0,3);
            \draw (5.2,0) -- (5.2,3);
            % Hacemos el A_{[i-1]} de arriba izquierda:
            \begin{scope}[shift={(0,3)}]
        		\filldraw[fill=purple] (-0.5-0.4,-0.5) -- (-0.5-0.4,0.5) -- (0.5+0.4,0.5) -- (0.5+0.4,-0.5) -- (-0.5-0.4,-0.5);
        		\draw (0,0) node {\scriptsize $A_{[i-1]}$};
        	\end{scope}
            % Hacemos el O_\alpha del medio izquierda:
            \begin{scope}[shift={(0,1.5)}]
        		\filldraw[fill=amaranth] (-0.5,-0.5) -- (-0.5,0.5) -- (0.5,0.5) -- (0.5,-0.5) -- (-0.5,-0.5);
        		\draw (0,0) node {\scriptsize $O_\alpha$};
        	\end{scope}
            % Hacemos el A^*_{[i-1]} de abajo izquierda:
            \begin{scope}[shift={(0,0)}]
        		\filldraw[fill=purple] (-0.5-0.4,-0.5) -- (-0.5-0.4,0.5) -- (0.5+0.4,0.5) -- (0.5+0.4,-0.5) -- (-0.5-0.4,-0.5);
        		\draw (0,0) node {\scriptsize $A^*_{[i-1]}$};
        	\end{scope}
            % Hacemos el E_{i,j}:
            \begin{scope}[shift={(2.7,1.5)}]
        		\filldraw[fill=purple] (-0.5-0.3,-0.5-1.5) -- (-0.5-0.3,0.5+1.5) -- (0.5+0.3,0.5+1.5) -- (0.5+0.3,-0.5-1.5) -- (-0.5-0.3,-0.5-1.5);
        		\draw (0,0) node {\scriptsize $\mathbb{E}_{i,j}$};
        	\end{scope}
            % Hacemos el A_{[j]} de arriba:
            \begin{scope}[shift={(5.2,3)}]
        		\filldraw[fill=purple] (-0.5-0.2,-0.5) -- (-0.5-0.2,0.5) -- (0.5+0.2,0.5) -- (0.5+0.2,-0.5) -- (-0.5-0.2,-0.5);
        		\draw (0,0) node {\scriptsize $A_{[j]}$};
        	\end{scope}
            % Hacemos el O_\beta' del medio:
            \begin{scope}[shift={(5.2,1.5)}]
        		\filldraw[fill=amaranth] (-0.5,-0.5) -- (-0.5,0.5) -- (0.5,0.5) -- (0.5,-0.5) -- (-0.5,-0.5);
        		\draw (0,0) node {\scriptsize $O'_\beta$};
        	\end{scope}
            % Hacemos el A^*_{[j]} de abajo izquierda:
            \begin{scope}[shift={(5.2,0)}]
        		\filldraw[fill=purple] (-0.5-0.2,-0.5) -- (-0.5-0.2,0.5) -- (0.5+0.2,0.5) -- (0.5+0.2,-0.5) -- (-0.5-0.2,-0.5);
        		\draw (0,0) node {\scriptsize $A^*_{[j]}$};
        	\end{scope}
            % Hacemos el Lambda_{[j]} de la derecha del todo:
            \begin{scope}[shift={(7,1.5)}]
        		\filldraw[fill=yellow] (-0.5-0.1,-0.5) -- (-0.5-0.1,0.5) -- (0.5+0.1,0.5) -- (0.5+0.1,-0.5) -- (-0.5-0.1,-0.5);
        		\draw (0,0) node {\scriptsize $\Lambda_{[j]}$};
        	\end{scope}
        \end{tikzpicture} \right. \nonumber
         \\
         &\hspace{3.5cm} \left. - \ 
        \begin{tikzpicture}[scale=.45, baseline={([yshift=-0.5ex]current bounding box.center)}, thick]
            % Hacemos las lineas:
            \draw (-1.5,0) -- (-1.5,3) -- (1.9,3) -- (1.9,0) -- (-1.5,0);
            \draw (0,0) -- (0,3);
            % Hacemos el A_{[i-1]} de arriba izquierda:
            \begin{scope}[shift={(0,3)}]
        		\filldraw[fill=purple] (-0.5-0.4,-0.5) -- (-0.5-0.4,0.5) -- (0.5+0.4,0.5) -- (0.5+0.4,-0.5) -- (-0.5-0.4,-0.5);
        		\draw (0,0) node {\scriptsize $A_{[i-1]}$};
        	\end{scope}
            % Hacemos el O_\alpha del medio izquierda:
            \begin{scope}[shift={(0,1.5)}]
        		\filldraw[fill=amaranth] (-0.5,-0.5) -- (-0.5,0.5) -- (0.5,0.5) -- (0.5,-0.5) -- (-0.5,-0.5);
        		\draw (0,0) node {\scriptsize $O_\alpha$};
        	\end{scope}
            % Hacemos el A^*_{[i-1]} de abajo izquierda:
            \begin{scope}[shift={(0,0)}]
        		\filldraw[fill=purple] (-0.5-0.4,-0.5) -- (-0.5-0.4,0.5) -- (0.5+0.4,0.5) -- (0.5+0.4,-0.5) -- (-0.5-0.4,-0.5);
        		\draw (0,0) node {\scriptsize $A^*_{[i-1]}$};
        	\end{scope}
            % Hacemos el Lambda_{[i-1]} de la derecha:
            \begin{scope}[shift={(1.9,1.5)}]
        		\filldraw[fill=yellow] (-0.5-0.4,-0.5) -- (-0.5-0.4,0.5) -- (0.5+0.4,0.5) -- (0.5+0.4,-0.5) -- (-0.5-0.4,-0.5);
        		\draw (0,0) node {\scriptsize $\Lambda_{[i-1]}$};
        	\end{scope}
        \end{tikzpicture}
        \ \cdot \
        \begin{tikzpicture}[scale=.45, baseline={([yshift=-0.5ex]current bounding box.center)}, thick]
            % Hacemos las lineas:
            \draw (3.75,0) -- (3.75,3) -- (7,3) -- (7,0) -- (3.75,0);
            \draw (5.2,0) -- (5.2,3);
            % Hacemos el A_{[j]} de arriba:
            \begin{scope}[shift={(5.2,3)}]
        		\filldraw[fill=purple] (-0.5-0.2,-0.5) -- (-0.5-0.2,0.5) -- (0.5+0.2,0.5) -- (0.5+0.2,-0.5) -- (-0.5-0.2,-0.5);
        		\draw (0,0) node {\scriptsize $A_{[j]}$};
        	\end{scope}
            % Hacemos el O_\beta' del medio:
            \begin{scope}[shift={(5.2,1.5)}]
        		\filldraw[fill=amaranth] (-0.5,-0.5) -- (-0.5,0.5) -- (0.5,0.5) -- (0.5,-0.5) -- (-0.5,-0.5);
        		\draw (0,0) node {\scriptsize $O'_\beta$};
        	\end{scope}
            % Hacemos el A^*_{[j]} de abajo izquierda:
            \begin{scope}[shift={(5.2,0)}]
        		\filldraw[fill=purple] (-0.5-0.2,-0.5) -- (-0.5-0.2,0.5) -- (0.5+0.2,0.5) -- (0.5+0.2,-0.5) -- (-0.5-0.2,-0.5);
        		\draw (0,0) node {\scriptsize $A^*_{[j]}$};
        	\end{scope}
            % Hacemos el Lambda_{[j]} de la derecha del todo:
            \begin{scope}[shift={(7,1.5)}]
        		\filldraw[fill=yellow] (-0.5-0.1,-0.5) -- (-0.5-0.1,0.5) -- (0.5+0.1,0.5) -- (0.5+0.1,-0.5) -- (-0.5-0.1,-0.5);
        		\draw (0,0) node {\scriptsize $\Lambda_{[j]}$};
        	\end{scope}
        \end{tikzpicture} \right| \nonumber \\
         &= \sup_{\substack{\alpha, \beta \\ \|\ket{\alpha}\|_2 = 1 \\ \| \ket{\beta} \|_2 = 1}}
         \frac{1}{\| O_\alpha \|_{\text{op}} \| O'_\beta \|_{\text{op}}} \left| \
         \begin{tikzpicture}[scale=.45, baseline={([yshift=-0.5ex]current bounding box.center)}, thick]
            % Hacemos las lineas:
            \draw (0.75,1) -- (0,1) -- (0,-1) -- (0.75,-1);
            % Hacemos el Lambda_{[j]} de la derecha:
            \begin{scope}[shift={(0,0)}]
        		\filldraw[fill=amaranth] (-0.5,-0.5) -- (-0.5,0.5) -- (0.5,0.5) -- (0.5,-0.5) -- (-0.5,-0.5);
        		\draw (0,0) node {\scriptsize $\alpha$};
        	\end{scope}
        \end{tikzpicture}
         \left( \
         \begin{tikzpicture}[scale=.45, baseline={([yshift=-0.5ex]current bounding box.center)}, thick]
            % Hacemos las lineas:
            \draw (-1.4,1) -- (1.4,1);
            \draw (-1.4,-1) -- (1.4,-1);
            % Hacemos el Lambda_{[j]} de la derecha:
            \begin{scope}[shift={(0,0)}]
        		\filldraw[fill=purple] (-0.5-0.2,-0.5-0.7) -- (-0.5-0.2,0.5+0.7) -- (0.5+0.2,0.5+0.7) -- (0.5+0.2,-0.5-0.7) -- (-0.5-0.2,-0.5-0.7);
        		\draw (0,0) node {\scriptsize $\mathbb{E}_{i,j}$};
        	\end{scope}
        \end{tikzpicture}
        \ \right. \right. \nonumber \\
        &\hspace{4.5cm} \left. \left. - \
        \begin{tikzpicture}[scale=.45, baseline={([yshift=-0.5ex]current bounding box.center)}, thick]
            \draw (-1.4,-1) -- (0,-1) -- (0,1) -- (-1.4,1);
            \draw (2.1,-1) -- (1.4,-1) -- (1.4,1) -- (2.1,1);
            \begin{scope}[shift={(0,0)}]
        		\filldraw[fill=yellow] (-0.5-0.4,-0.5) -- (-0.5-0.4,0.5) -- (0.5+0.4,0.5) -- (0.5+0.4,-0.5) -- (-0.5-0.4,-0.5);
        		\draw (0,0) node {\scriptsize $\Lambda_{[i-1]}$};
        	\end{scope}
        \end{tikzpicture}
         \ \right)
         \begin{tikzpicture}[scale=.45, baseline={([yshift=-0.5ex]current bounding box.center)}, thick]
            % Hacemos las lineas:
            \draw (-0.75,1) -- (0,1) -- (0,-1) -- (-0.75,-1);
            % Hacemos el Lambda_{[j]} de la derecha:
            \begin{scope}[shift={(0,0)}]
        		\filldraw[fill=amaranth] (-0.5,-0.5) -- (-0.5,0.5) -- (0.5,0.5) -- (0.5,-0.5) -- (-0.5,-0.5);
        		\draw (0,0) node {\scriptsize $\beta$};
        	\end{scope}
        \end{tikzpicture}
        \ \right| \nonumber 
         \\
         &= \sup_{\substack{\alpha, \beta \\ \|\ket{\alpha}\|_2 = 1 \\ \| \ket{\beta} \|_2 = 1}}
         \frac{| \mel{\alpha}{(\mathbb{E}_{i,j} - \dyad{\Lambda_{[i-1]}}{\mathds{1}})}{\beta} |}{\| O_\alpha \|_{\text{op}} \| O'_\beta \|_{\text{op}}}
         \label{eq:decay_4} 
     \end{align}
     
     We now define operators $\mathcal{P}: \alpha \mapsto O_\alpha$, $\mathcal{P'}: \beta \mapsto O_\beta'$, with the associated operator norms
     \begin{equation*}
         \| \mathcal{P} \|_{\text{op}\to\text{op}} = \sup_\alpha \left\{ \frac{\|O_{\alpha}\|_{\text{op}}}{\|\alpha\|_{\text{op}}} \right\}, \
         \| \mathcal{P'} \|_{\text{op}\to\text{op}} = \sup_\beta \left\{ \frac{\|O_{\beta}'\|_{\text{op}}}{\|\beta\|_{\text{op}}} \right\}.
     \end{equation*}
     Note that for any $\alpha, \beta$ with $\| \ket{\alpha} \|_2 = \|\ket{\beta}\|_2 = 1$, and using the fact that $\| \alpha \|_{\text{op}} \leq \|\ket{\alpha} \|_2$, we have that
     \begin{equation} \label{eq:decay_5}
         \frac{\|\mathcal{P}\|_{\text{op}\to\text{op}}}{\|O_\alpha\|_{\text{op}}} \geq 1, \quad 
         \frac{\|\mathcal{P}'\|_{\text{op}\to\text{op}}}{\|O'_\beta\|_{\text{op}}} \geq 1 .
     \end{equation}
     We are now going to provide an upper bound on the desired quantity $\Delta_{i,j}$, where
     \begin{align*}
         \Delta_{i,j} &:= \|\mathbb{E}_{i,j} - \dyad{\Lambda_{[i-1]}}{\mathds{1}}\|_{\text{op}} \\
         &= \sup_{\substack{\alpha, \beta \\ \|\alpha\|_2 = 1 \\ \|\beta\|_2 = 1}} |\mel{\alpha}{(\mathbb{E}_{i,j} - \dyad{\Lambda_{[i-1]}}{\mathds{1}})}{\beta}|,
     \end{align*}
     using Eq. \eqref{eq:decay_4} together with Eq. \eqref{eq:decay_5}. That is, letting $M := \| \mathcal{P} \|_{\text{op}\to\text{op}} \| \mathcal{P}' \|_{\text{op}\to\text{op}}$, we have
     \begin{align*}
         \Delta_{i,j}
         &=  
         \sup_{\substack{\alpha, \beta \\ \|\alpha\|_2 = 1 \\ \|\beta\|_2 = 1}} M \cdot \frac{ | \mel{\alpha}{(\mathbb{E}_{i,j} - \dyad{\Lambda_{[i-1]}}{\mathds{1}})}{\beta} |}{\|O_\alpha\|_{\text{op}} \|O'_\beta\|_{\text{op}}} \cdot 
         \\
         &\hspace{3cm} \cdot \frac{\|O_\alpha\|_{\text{op}} \|O'_\beta\|_{\text{op}}}{\| \mathcal{P} \|_{\text{op}\to\text{op}} \| \mathcal{P}' \|_{\text{op}\to\text{op}}} \\
         &\leq  
         \sup_{\substack{\alpha, \beta \\ \|\alpha\|_2 = 1 \\\|\beta\|_2 = 1}} M \cdot \frac{| \mel{\alpha}{(\mathbb{E}_{i,j} - \dyad{\Lambda_{[i-1]}}{\mathds{1}})}{\beta} |}{\|O_\alpha\|_{\text{op}} \|O'_\beta\|_{\text{op}}} \\
         &\leq M K e^{-|i-j|/\xi}
     \end{align*}

     Let us further upper bound the quantity $M$ by studying the value $\| \mathcal{P} \|_{\text{op}\to\text{op}} = \sup_{\|\alpha\|_{op} = 1} \{ \|O_\alpha\|_{op}\}$. Note that for any $\alpha$ with $\|\alpha\|_{op} = 1$, we have
     \begin{align*}
        \| O_\alpha \|_{\text{op}} &= \sup_{\substack{\ket{\psi},\ket{\phi} \\ \|\ket{\psi}\|_2 = \| \ket{\phi} \|_2 = 1}} \left| \mel{\psi}{O_\alpha}{\phi} \right| \\
        &\leq \sup_{\ket{\tilde{\psi}}, \ket{\tilde{\phi}}} |\langle \tilde{\psi} | \Lambda_{[i-1]}\otimes\alpha | \tilde{\phi} \rangle | \\
        &\leq \| |\tilde{\psi}\rangle \|_2 \| |\tilde{\phi} \rangle \|_2 \| \Lambda_{[i-1]} \otimes \alpha \|_{\text{op}} \\
        &= \| |\tilde{\psi}\rangle \|_2 \| |\tilde{\phi}\rangle \|_2 \underbrace{\| \Lambda_{[i-1]} \|_{\text{op}}}_{\leq 1} \underbrace{\| \alpha \|_{\text{op}}}_{=1}
        \\
        &\leq \| |\tilde{\psi}\rangle \|_2 \| |\tilde{\phi}\rangle \|_2 \\
        &\leq 
        \| \ket{\psi} \|_2 \| \ket{\phi} \|_2 \|A_{[i-1]}^{-1}\|_{\text{op}} \|(A_{[i-1]}^*)^{-1}\|_{\text{op}} \\
        &= \|A_{[i-1]}^{-1}\|_{\text{op}}^2
    \end{align*}
     where we have used the fact that $\| \Lambda_{[i-1]} \|_{\text{op}} \leq \Tr[\Lambda_{[i-1]}] = 1$. Similarly, 
     \begin{equation*}
         \| \mathcal{P}' \|_{\text{op}\to\text{op}} \leq \|A_{[j]}^{-1}\|_{\text{op}}^2
     \end{equation*}
     Therefore, 
     \begin{equation*}
         \|\mathbb{E}_{i,j} - \dyad{\Lambda_{[i-1]}}{\mathds{1}}\|_{\text{op}} \leq 
         C e^{-|i-j|/\xi}
     \end{equation*} 
     as desired, where $C := K \max_{i} \|A_{[i]}^{-1}\|_{\text{op}}^4 = O(1)$ by our assumption.
\end{proof}

\subsection{The product structure of ergodic non-TI MPS-up}
\label{app:proof_nonTI_genericMPS}

Now, we proceed to show that ergodicity and the MPS-up property imply that the state factorizes into a product of almost all sites or blocks of a few sites. As already mentioned, we may assume without loss of generality that the tensors are injective. 

%\nonTIpurity*
\begin{restatable}{proposition}{nonTIpurity}
    \label{prop:purity_non-TI_injective}    
    Consider an ergodic family of injective non-TI MPS $\{\ket{\psi_N}\}$ defined by normalised tensors $A^{[n]}$ as
    \begin{equation*}
        \ket{\psi_N} :=
        \begin{tikzpicture}[scale=.6, baseline={([yshift=-1ex]current bounding box.center)}, thick]
            \begin{scope}[shift={(0,0)}]
        		\draw[shift={(0,0)},dotted] (0.5,0) -- (4,0);
                \MPSTensorrect{(0,0)}{$A^{[1]}$}{purple}{0.7}
                \MPSTensorrect{1.5,0}{$A^{[2]}$}{purple}{0.7}
                \MPSTensorrect{4.5,0}{$A^{[N]}$}{purple}{0.7}
                \draw (-1,0) -- (-1,-0.8*0.7) -- (5.5,-0.8*0.7) -- (5.5,0);
        	\end{scope}
        \end{tikzpicture}
    \end{equation*}
    with the MPS-up$_{\varepsilon, D}$ property for all $N$. If $\varepsilon < \frac{1}{4D}$, then it must necessarily consist of states that are products of almost all sites.
\end{restatable}
\begin{proof}
    We follow the steps in the proof of Proposition \ref{prop:purity_normal}. For any $k, n \in \mathbb{N}$, let $S_{k,n}^N := \{i_1(k), i_2(k), \dots, i_n(k)\} \subseteq \{1, \dots, N\}^n$ be a strictly increasing set of natural numbers, the separation between them growing arbitrarily large as $k$ increases, i.e. the functions $i_\alpha(k)$ satisfy that $i_\alpha(k) < i_{\alpha+1}(k)$ and $i_{\alpha+1}(k) - i_\alpha(k) \to \infty$ as $k\to \infty$ for all $\alpha$. 
    
    For each $N$ and each set $S_{k,n}^N$, we define $\pi \in \mathcal{S}_{N}$ as the permutation sending all particles labelled in increasing order by $S_{k,n}^N$ to the beginning of the chain. That is, $\pi(i_\alpha(k)) = \alpha$ for $\alpha = 1, \dots, n$, and for the rest of the particles labelled in increasing order as $\{m_1, \dots, m_{N-n}\} = \{1, \dots, N\} \setminus S_{k,n}^N$, we have $\pi(m_i) = n+i$. The permutation used in the proof of Prop. \ref{prop:purity_normal} corresponds to set $S^N_{k,n} = \{k, 2k, \dots, nk\}$.

    As was argued in Proposition \ref{prop:purity_normal}, a lower bound on the purity of subsystem $S^N_{k,n}$ can be obtained through the MPS-up$_{\varepsilon, D}$ property, implying that
    \begin{equation*}
        \Tr[(\rho_{S})^2] \geq \frac{1}{D} - 4\varepsilon.
    \end{equation*}

    On the other hand, we can express this purity as follows
    \begin{widetext}
    \begin{align*}
    \Tr[(\rho_{S})^2] = \frac{1}{Z_N} \Tr[ \prod_{\alpha=1}^n 
    \begin{array}{c}
		\begin{tikzpicture}[scale=.45, baseline={([yshift=-0ex]current bounding box.center)}, thick]
            \ETensor{0,0}{$\mathbb{E}_{i_{\alpha-1}(k)+1, i_\alpha(k)}$}{purple}{2.5}{0.7}
            \ETensor{0,1.4}{$\mathbb{E}_{i_{\alpha-1}(k)+1, i_\alpha(k)}$}{purple}{2.5}{0.7}
            % Draw the swaps:
            \draw (1.7+1.3,0) -- (1.7+1.8,1.4);
            \draw (1.7+1.3,1.4) -- (1.7+1.8,0);
            \draw (1.7+1.3,0.7) -- (1.7+1.8,0.7);
            \draw (1.7+1.3,2.1) -- (1.7+1.8,2.1);
            % Draw the second batch of E tensors:
            \ETensor{2.3+2.8,0}{$\mathbb{E}_{[i_\alpha(k)]}$}{purple}{1.1}{0.7}
            \ETensor{2.3+2.8,1.4}{$\mathbb{E}_{[i_\alpha(k)]}$}{purple}{1.1}{0.7}
            % Draw second batch of swaps:
            \draw (2.9+3.8,0) -- (2.9+4.3,1.4) -- (2.9+5,1.4);
            \draw (2.9+3.8,1.4) -- (2.9+4.3,0) -- (2.9+5,0);
            \draw (2.9+3.8,0.7) -- (2.9+5,0.7);
            \draw (2.9+3.8,2.1) -- (2.9+5,2.1);
            % Now draw parentheses:
            \draw [decorate, decoration = {calligraphic brace,mirror}] (-2.7,2.3) --  (-2.7,-0.2);
            \draw [decorate, decoration = {calligraphic brace}] (3.2+4.4,2.3) --  (3.2+4.4,-0.2);
        \end{tikzpicture}
    \end{array} ]
    \xrightarrow{\scriptsize k\to\infty} \nonumber 
    \prod_{\alpha=1}^n
    \begin{array}{c}
		\begin{tikzpicture}[scale=.45, baseline={([yshift=0ex]current bounding box.center)}, thick]
            % The E tensors and their unions on the left:
            \ETensor{0,0}{$\mathbb{E}_{[i_\alpha(k)]}$}{purple}{1.1}{0.7}
            \ETensor{0,1.4}{$\mathbb{E}_{[i_\alpha(k)]}$}{purple}{1.1}{0.7}
            \draw (-1-0.6,0.7) -- (-1-0.6,1.4);
            \draw (-1-0.6,0) -- (-1.3-0.6,0) -- (-1.3-0.6,2.1) -- (-1-0.6,2.1);
            % Draw the swaps:
            \draw (1+0.6,0) -- (1.5+0.6,1.4);
            \draw (1+0.6,1.4) -- (1.5+0.6,0);
            \draw (1+0.6,0.7) -- (1.5+0.6,0.7);
            \draw (1+0.6,2.1) -- (1.5+0.6,2.1);
            % Draw the right fixed points:
            \RightVector{2.5+1.4,0}{$\sigma_{i_\alpha(k)+1}$}{yellow}{1.3}{0.7}
            \RightVector{2.5+1.4,1.4}{$\sigma_{i_\alpha(k)+1}$}{yellow}{1.3}{0.7}
        \end{tikzpicture}
    \end{array}
     =: \prod_{\alpha=1}^n \eta_{[i_\alpha(k)]},
    \end{align*}
    \end{widetext} 
    where $Z_N := \Tr[\prod_{j=1}^N \mathbb{E}_{[i]}]^2 \xrightarrow{k\to\infty} 1$, and we used assumption (\ref{eq:condition_non-TI_purity_proof}) in the limit $k \to \infty$. Then, we have that
    \begin{align*} 
        \Tr[(\rho_S)^2] &\xrightarrow{k\to\infty} \prod_{\alpha=1}^n \eta_{[i_\alpha(k)]} \geq \frac{1}{D} - 4\varepsilon \label{eq:purity2}, 
    \end{align*}
    where the lower bound is a strictly positive quantity if $\varepsilon < \frac{1}{4D}$. Noting that $0 < \eta_{[i_\alpha(k)]} \leq 1$, and using an argument analogous to the one in the proof of Proposition \ref{prop:purity_normal}, this leads to a contradiction for sufficiently large $n$, unless $D_{[i_\alpha(k)]} = 1$ for all $\alpha$, except for (at most) a finite number of them. Since we can choose $S_{k,n}^N$ arbitrarily, we conclude that the bond dimension of almost all of the tensors $\{A^{[i]}\}$ is necessarily one, except maybe for $A^{[i]}$ with $i \in T$ for some subset $T \subseteq \mathbb{N}$ of finite size, $|T| < \infty$. Therefore, the family necessarily consists of states that are products of almost all sites and the claim follows. 
\end{proof}

\section{Proofs for the MPS to product-type ansatz replacement} \label{sec:app_algorithms}

\begin{lemma} \label{lemma:technical_algorithms}
    Consider a Hamiltonian $\mathcal{H}$ with a gap $\Delta > 0$ on $N$ particles, with ground state energy 0 wlog, and a unique ground state $\ket{\psi_0}$. If $\ket{\psi_1}, \ket{\psi_2}$ are any two states with energies $\epsilon_1, \epsilon_2$ close to 0, then
    \begin{equation*}
        \| \ket{\psi_1} - \ket{\psi_2} \|_2 \leq 2 \sqrt{\frac{\max_i \epsilon_i}{\Delta}}
    \end{equation*}
\end{lemma}
\begin{proof}
    Wlog, we can write
    \begin{align*}
        \ket{\psi_1} &:= \sqrt{\delta_1} \ket{\psi_0} + \sqrt{1 - \delta_1} \ket{\xi_1}, \\
        \ket{\psi_2} &:= \sqrt{\delta_2} \ket{\psi_0} + \sqrt{1 - \delta_2} \ket{\xi_2},
    \end{align*}
    where $\delta_i \in [0,1]$, $\ket{\psi_0}, \ket{\xi_i}$ are normalized, and $\braket{\xi_i}{\psi_0} = 0$. We note that
    \begin{align*}
        \epsilon_i = \expval{\mathcal{H}}{\psi_i} &= (1-\delta_i) \expval{\mathcal{H}}{\xi_i} \geq (1-\delta_i) \Delta \\
         \implies & 1 - \delta_i \leq \frac{\epsilon_i}{\Delta}
    \end{align*}
    From this, we obtain
    \begin{align*}
        \braket{\psi_1}{\psi_2} &= \sqrt{\delta_1 \delta_2} + \sqrt{(1-\delta_1)(1-\delta_2)} \braket{\xi_1}{\xi_2} \\
        \text{Re}(\braket{\psi_1}{\psi_2}) &\geq 
        \sqrt{\delta_1 \delta_2} - \sqrt{(1-\delta_1)(1-\delta_2)} \\
        &\geq 2 \min_i \delta_i - 1 \geq 1 - \frac{2}{\Delta} \max_i \epsilon_i
        \\
        \| \ket{\psi_1} - \ket{\psi_2} \|_2^2 &= 
        2 \left(1 - \text{Re}(\braket{\psi_1}{\psi_2})\right) \leq \\
        &\leq 4 \frac{\max_i \epsilon_i}{\Delta}.
    \end{align*}
\end{proof}

\begin{mycly}{5} 
    Under the above assumptions, the true ground state $\ket{\psi_0}$ of $\mathcal{H}$ is at most $2\sqrt{\varepsilon/\Delta}$ away from one of the simple ansätze listed in the box of Table \ref{table:results} that corresponds to the structure assumed on $\ket{\psi_1}$, provided $\varepsilon/\Delta$ is small enough (i.e. $\ket{\psi_0}$ is a product state, a small superposition thereof, or a product of short-range entangled blocks).
\end{mycly}
\begin{proof}
    Assume that $\ket{\psi_1}$ and $\ket{\psi_2}$ have energies $\varepsilon_1, \varepsilon_2$, respectively, where $\varepsilon_i \leq \varepsilon$. Using Lemma \ref{lemma:technical_algorithms}, the assumption can be expressed as
    \begin{equation} \label{eq:corollary_aux}
        \| U_\pi \ket{\psi_1} - \ket{\psi_2} \| \leq 2 \sqrt{\frac{\max_i \varepsilon_i}{\Delta}} =: \tilde{\varepsilon} 
    \end{equation}    
    where $\ket{\psi_1}$ would be the MPS state on the particles arranged as $\{1, \dots, N\}$, and $\ket{\psi_2}$ the other MPS state found when arranging the particles as $\{\pi(1), \dots, \pi(N)\}$. Then, $\ket{\psi_2}$ has the MPS-up property with respect to a particular permutation $\pi$ and we can use Table \ref{table:results} if $\tilde{\varepsilon}$ is sufficiently small. Note that it is enough to require that Eq. (\ref{eq:corollary_aux}) holds for a single specific permutation. Indeed, for the exact MPS-up case, any $\pi$ where the Schmidt rank computed in Eq. (\ref{eq:rank_counting_normal_proof_aux}) is larger than $D$ suffices. For the approximate MPS-up case, choosing $\pi$ to be the permutation in Prop. \ref{prop:purity_normal} or Prop. \ref{prop:purity_non-TI_injective} is also enough.

    Applying Lemma \ref{lemma:technical_algorithms} again, knowing that $\ket{\psi_1}, \ket{\psi_2}$ are product states or superpositions of a few of them, we have
    \begin{equation*}
       \| \ket{\psi_i} - \ket{\psi_0} \| \leq 2 \sqrt{\frac{\epsilon_i}{\Delta}}
    \end{equation*}
    for $i = 1, 2$. Therefore, we can conclude that the ground state $\ket{\psi_0}$ is $2\sqrt{\min_i \epsilon_i/\Delta}$-close to the claimed ansätze.
 \end{proof}

\section{Lower bound for the efficiency of TI MPS representations of MPS-up with increasing tensor rank}
\label{sec:app_DTI_bound}

Let $\{| \psi_N(A_{(N)}) \rangle\}_N$ be a family of TI MPS defined as
\begin{equation*}
        |\psi_N(A_{(N)})\rangle :=
        \begin{tikzpicture}[scale=.6, baseline={([yshift=-1ex]current bounding box.center)}, thick]
            \begin{scope}[shift={(0,0)}]
        		\draw[shift={(0,0)},dotted] (0.5,0) -- (4,0);
                \MPSTensorrect{(0,0)}{$A_{(N)}$}{purple}{0.8}
                \MPSTensorrect{1.5,0}{$A_{(N)}$}{purple}{0.8}
                \MPSTensorrect{4.5,0}{$A_{(N)}$}{purple}{0.8}
                \draw (-1,0) -- (-1,-0.8*0.8) -- (5.5,-0.8*0.8) -- (5.5,0);
        	\end{scope}
        \end{tikzpicture} 
        \in (\mathbb{C}^{\otimes d})^{\otimes N}.
    \end{equation*}
For each value of $N$, the site matrices $A_{(N)}$ can be expressed in canonical form as shown in Eq. (\ref{eq:CF_def}) in terms of a basis of normal tensors $\{A_{j,(N)}\}$ with $b_{(N)}$ elements, $A^i_{j,(N)} \in \mathcal{M}_{D_{j,(N)}}(\mathbb{C})$. 

Even though every TI state on a finite chain of $N$ particles admits a TI MPS representation, the bond dimension can generally increase with the system size \cite{Perez-Garcia2007}. Actually, the problem of determining what is the minimal TI MPS for a given state is an open question \cite{barthel2022closedness, klimov2023translation}. 

Here we show that the existing lower bound for the bond dimension of any TI MPS representation of the W-state, of $\Omega(N^{1/(3+\delta)})$ for each $\delta > 0$ \cite{Perez-Garcia2007, Michalek2018}, is also valid for families of exact MPS-up with increasing tensor rank, by closely following the proof of Corollary A.1 in \cite{Perez-Garcia2007}. We restate Corollary \ref{cor:lowerbound_DTI} below for the reader's convenience.

\lowerboundDTI*
\begin{proof}
    Assume that $N/2 > L_{BI}^{(N)} := 3(b_{(N)} - 1)(L_0^{(N)} + 1)$, where $L_0^{(N)}$ is the maximum injectivity length over all the blocks in the BNT, meaning that $L_0^{(N)} \leq \max_j 2(D_{j,(N)})^2 (6 + \log_2 D_{j,(N)})$. Suppose as well that we have already blocked the physical sites to remove periodicities, so $p = 1$. Given these conditions, we can partition our state into two parts containing particles $\{1, \dots, R\}$ and $\{R+1, \dots, N\}$ with $R > L^{(N)}_{BI}$, where $L^{(N)}_{BI}$ is the block-injectivity length of $A_{(N)}$ defined in section \ref{sec:TNs-background}. 
    
    Block-injectivity implies that the sets of vectors $ \{ |\Psi^{(R),N}_{j,\alpha,\beta}\rangle\}_{j,\alpha, \beta}$ and $ \{ |\Psi^{(N-R),N}_{j,\alpha,\beta}\rangle \}_{j,\alpha, \beta}$ with $j \in \{1, \dots, b_{(N)}\}$, $\alpha, \beta \in \{1, \dots, D_{j,(N)}\}$, and
    \begin{equation*}
       |\Psi^{(m),N}_{j, \alpha, \beta}\rangle := \sum_{i_{1}, \dots, i_{m}=1}^d \mel{\alpha}{A_{j,(N)}^{i_1} \dots A_{j,(N)}^{i_m}}{\beta} \ket{i_1 \dots i_R}
    \end{equation*}
    are linearly independent, and thus have dimension $\sum_{j=1}^{b_{(N)}} D_{j,(N)}^2$. This quantity is therefore the rank of the reduced density matrix corresponding to subsystem $\{1, \dots, R\}$, which is equal to the Schmidt rank across that bipartition. According to the MPS-up$_{0,D}$ property, the Schmidt rank should be upper bounded by $D$, so we have
    \begin{equation*}
        \sum_{j=1}^{b_{(N)}} D_{j,(N)}^2 = O(1) \implies b_{(N)}, D_{j,(N)} = O(1), \ \forall j, N.
    \end{equation*}
    If this was the case, then $L_{BI} = O(1)$, and we could apply Theorem \ref{thm:exact_MPS-up_thm} for any sufficiently large $N$ such that $N > L_{BI} (2\log_2 D + 1)$. This would mean that $\ket{\psi_N(A^{(N)})}$ can be written as a superposition of $b_N = O(1)$ product states, which imposes a constant upper bound on the tensor rank. However, this would contradict the assumption that the tensor rank increases with $N$. Therefore, we necessarily have that $N/2 \leq 3(b_{(N)} - 1)(L_0 + 1) = O(D_{(N)}^3 \log D_{(N)})$, which implies that for each $\delta > 0$, $D_{(N)} = \Omega(N^{1/(3+\delta)})$.    
\end{proof}

\end{document}